\documentclass[12pt]{article}
\usepackage{amsmath}
\usepackage{graphicx,psfrag,epsf}
\usepackage{enumerate}
\usepackage[figuresright]{rotating}
\usepackage{natbib}
\usepackage{amssymb}
\usepackage{pifont}
\usepackage{multirow}
\usepackage{booktabs}
\usepackage{url} 
\usepackage{adjustbox}
\newcommand{\blind}{0}
\newcommand\bbeta{\boldsymbol{\beta}}
\newcommand\pphi{\boldsymbol{\phi}}
\newcommand\MM{\mathbf{M}}

\newcommand\XX{\mathbf{X}}
\newcommand\x{\mathbf{x}}
\newcommand\z{\mathbf{z}}
\newcommand\bv{\mathbf{v}}
\newcommand\hbv{\hat{\mathbf{v}}}
\newcommand\bw{\mathbf{w}}
\newcommand\X{\boldsymbol{X}}
\newcommand\Y{\boldsymbol{y}}
\newcommand\Z{\boldsymbol{Z}}
\newcommand\cvp{\stackrel{p}{\longrightarrow}}
\newcommand\cvd{\stackrel{d}{\longrightarrow}}

\newcommand\Stij{S_{T,i,j}}
\newcommand\hStij{\hat{S}_{T,i,j}}

\newcommand\logn{\log n}

\newcommand\R{\mathbf{R}}

\newcommand\vv{\hat{v}}

\newcommand\VVV{\hat{V}}
\newcommand\hab{\hat{\bbeta}}

\newcommand\rhohat{\hat{\rho}_{ij}}

\newcommand\inn[2]{\langle #1, #2 \rangle}

\usepackage{amsmath}
\usepackage{tcolorbox}
\usepackage{amssymb}
\usepackage{amsfonts}
\usepackage{multirow}
\usepackage{amsthm}
\usepackage{mathrsfs}
\usepackage{indentfirst}
\newtheorem{theorem}{Theorem}
\newtheorem{lemma}{Lemma}

\theoremstyle{definition}

\newtheorem{remark}{Remark}
\newtheorem{assumption}{Assumption}
\usepackage{fancybox,framed}

\usepackage[utf8]{inputenc}
\usepackage{graphicx}

\usepackage{float}
\begin{document}

\def\spacingset#1{\renewcommand{\baselinestretch}%
{#1}\small\normalsize} \spacingset{1}


\if0\blind
{
  \title{\bf Unified and robust Lagrange multiplier type tests for cross-sectional independence in large panel data models}
  \author{Zhenhong Huang\footnote{ Department of Statistics and Actuarial Science, The University of Hong Kong. Email: zhhuang7@connect.hku.hk}, \;
    Zhaoyuan Li\footnote{School of Data Science, The Chinese University of Hong Kong, Shenzhen. Email: lizhaoyuan@cuhk.edu.cn}
    \;and\; Jianfeng Yao\footnote{School of Data Science, The Chinese University of Hong Kong (Shenzhen). Email: jeffyao@cuhk.edu.cn}\hspace{.2cm} \\}
    \date{}
  \maketitle
} \fi

\if1\blind
{
  \bigskip
  \bigskip
  \bigskip
  \begin{center}
    {\LARGE\bf Title}
\end{center}
  \medskip
} \fi

\bigskip
\begin{abstract}
\noindent This paper revisits the Lagrange multiplier type test for the null hypothesis of no cross-sectional dependence in large panel data models. We propose a unified test procedure and its power enhancement version, which show robustness for a wide class of panel model contexts. Specifically, the two procedures are applicable to both heterogeneous and fixed effects panel data models with the presence of weakly exogenous as well as lagged dependent regressors, allowing for a general form of non-normal error distribution. With the tools from Random Matrix Theory, the asymptotic validity of the test procedures is established under the simultaneous limit scheme 
where the number of time periods ($T$) and the number of cross-sectional units ($n$) go to infinity proportionally. The derived theories are accompanied by detailed Monte Carlo experiments, which confirm the robustness of the two tests and also suggest the validity of the power enhancement technique.  
\end{abstract}

\noindent%
{\it Keywords:}  Cross-sectional dependence, Large panels, Coefficient heterogeneity, Weak exogenous, Random Matrix Theory
\vfill

\newpage
\spacingset{1.8} 
\section{Introduction}
\label{sec:intro}
In panel data analysis, the problem of error cross-sectional dependence 
has attracted substantial attention in recent years. This  cross-sectional dependence can arise for various reasons, such as omitted common or spatial effects. Ignoring error cross-sectional dependence can have dramatic effects on conventional panel estimators, e.g., the least squares, fixed and random effect estimators, and yield invalid inferential procedures such as commonly used panel unit root tests, where these tests assume cross-sectional independence. Therefore, designing efficacious tests for  cross-sectional dependence is essential in panel data analysis.

There has been much work on testing for cross-sectional dependence in the literature. \cite{breusch1980lagrange} proposed a Lagrange multiplier ($LM$) test based on the squared pair-wise Pearson correlation coefficients of the residuals. Under the null hypothesis of no cross-sectional dependence,  the $LM$ test is asymptotically chi-squared distributed with $T\rightarrow\infty$ and $n$  fixed. However, it is not applicable for large $n$, which renders its popularity considering recent researches have focused on the large panels where both $T$ and $n$ can be large. In such high dimensional setting, there are two mainstream schemes considered by statisticians and econometricians, being known as sequential limit scheme and simultaneous limit scheme defined as following:
$$\text{SEQ-L:} \; T\rightarrow\infty, \; \text{followed by} \; n\rightarrow\infty,$$
and
$$\text{SIM-L:} \; (n, T) \rightarrow\infty \; \text{such  that} \; \frac{n}{T}\rightarrow c \in (0,\infty), $$
respectively. Under the SIM-L scheme, \cite{frees1995assessing} proposed a distribution free $LM$ type test allowing for large $n$, $R^2_{AVE}$,  based on the squared pair-wise Spearman rank correlation coefficients, which is asymptotically distributed as Chi-squared.  The test has however imposed limitations on the number of regressors and could be oversized for small $T$.  \cite{pesaran2004general} suggested a scaled version of the $LM$ test, denoted by $CD_{LM}$, and showed its asymptotic property. The author however pointed out that the $CD_{LM}$ test is not correctly centered at zero for small $T$, and is likely to exhibit large size distortions as $n$ increases. \cite{pesaran2004general} also proposed an alternative approach, the $CD_P$ test, which employs pair-wise Pearson correlation coefficients of the residuals as well, but without squaring them. This test has universally correct size under a broad course of panel data model designs, but it lacks power when correlation coefficients within panel units have variable signs leading to certain cancellation effect. Particularly, this happens when the errors are generated from a factor model where the loadings average to zero. \cite{pesaran2015testing} extended the $CD_P$ test to the scenario of weak cross-sectional dependence.  \cite{pesaran2008bias} put forward another approach, $LM_{adj}$, by deriving the exact expected values and variances of the squared correlation coefficients under the assumption of normally distributed errors and strictly exogenous regressors. By applying the classical central limit theory, it is shown that the $LM_{adj}$ test converges to standard normal distribution under the SEQ-L scheme. 
 \cite{bailey2021lagrange} proposed a new $LM$ type test, $LM_{RMT}$, and proved its asymptotic normality under the SIM-L scheme using Random Matrix Theory. However, they require the assumptions of normal regressors and normal errors.
In a slope homogeneity setting, \cite{baltagi2012lagrange} analyzed the performance of the $CD_{LM}$ test in the fixed effects panel data model, then presented a bias corrected test, $LM_{bc}$, and established its asymptotic normality under the SIM-L scheme assuming normal errors and strictly exogenous regressors. \cite{baltagi2012lagrange} also showed that the $LM_{bc}$ test can be applied to dynamic panel data model with fixed effects, using the within estimator proposed by \cite{hahn2002asymptotically}. 
 However, the slope homogeneity restriction has often been rejected in empirical analyses, see a detailed survey by \cite{baltagi2008pool}. Meanwhile, there are limited tests designed for the dynamic panel data model.  Examples include the GMM approach of \cite{sarafidis2009test} applied to panels with homogeneous coefficients under factor model representation, and a heteroskedasticity robust $LM$ test, $LM_{HOY}$, of \cite{halunga2017heteroskedasticity} with $n^2/T\rightarrow 0$ required. 

We make two distinct contributions in this paper. First, we propose a $LM$ type test statistic that can be applied to a wide class of linear panel models. Specifically, we show that our test is robust to both  static and dynamic heterogeneous panel data models with weakly exogenous regressors and non-normal errors.
With tools from Random Matrix Theory, we treat sample correlation matrix of residuals directly as a random matrix to establish the asymptotic normality of the test statistics under the SIM-L scheme, which has been suggested to be a more reliable strategy when dealing with high-dimensional statistical problems, see \cite{yao2015sample}. We also show that the proposed test is mathematically equivalent to both the $LM_{RMT}$ and the $LM_{bc}$ test. This finding theoretically enriches the $LM_{RMT}$ and $LM_{bc}$ test by relaxing the restrictive assumptions they need. It is worth mentioning that the existing literature on testing for cross-sectional dependence has mostly focused on the case of strictly exogenous regressors, including the $LM_{adj}$, $LM_{bc}$ and $LM_{HOY}$ tests. Though \cite{pesaran2015testing} showed that the $CD_P$ test is also applicable to autoregressive panel data models so long as the errors are symmetrically distributed, the properties of the $CD_P$ test for dynamic panels that include weakly exogenous regressors have not yet been investigated. The new proposed test fills this gap for the weakly exogenous regressors case, and to the best of our knowledge, there is no such a unified test so far. 

Second, weak cross-sectional dependence is common in empirical applications, see, for example, \cite{bailey2016exponent}, \cite{ertur2017weak}. This leads to a sparse correlation structure with few nonzero off-diagonal entries. Therefore, it is important to test the existence of such weak cross-sectional dependence, which corresponds to sparse alternatives in the high-dimensional statistics literature. The mainstream of more powerful tests for sparse alternatives are based on the maxima absolute value of sample correlations, see  \cite{cai2011limiting} and \cite{hall2010innovated}. However, these tests require stringent conditions that are `unfeasible' in econometric applications and often suffer from size distortions due to slow rates of convergence. In view of this,  we propose a novel and easy-implemented test statistic for cross-sectional dependence based on the fourth moment of sample correlation to boost the power of detecting sparse alternatives. Again, this test can be used in any aforementioned panel setting  under suitable conditions on higher moments of the errors.

The remainder of the paper is organized as follows. Section \ref{sec:existing} discusses the existing tests for no cross-sectional dependence.  Section \ref{sec:test} introduces the new test statistics and establishes the limiting distributions under the SIM-L scheme. Sections \ref{sec:dyanamic} and \ref{sec:fixed} demonstrate that the proposed tests can be extended to both the dynamic and fixed effects panel data models. Section \ref{sec:simulation} reports the results of  Monte Carlo simulations. Section \ref{sec:conclusion} provides concluding remarks and further discussions.

\textbf{Notations.} Throughout the paper, for a matrix $\mathbf{A} \in \mathcal{R}^{n\times n}$,  $tr(\mathbf{A})$ represents the trace $\mathbf{A}$. We use $\lambda_1(\mathbf{A}), \dots, \lambda_n(\mathbf{A})$ to denote $n$ eigenvalues of $\mathbf{A}$. Further, for vectors $\mathbf{a}, \mathbf{b} \in \mathcal{R}^{n\times 1}$, we write  $\inn{\mathbf{a}}{\mathbf{b}}=\mathbf{a}^{T}\mathbf{b}$ for their scalar product. In addition. $\|\cdot\|$ represents the Euclidean norm for a vector and the induced operator norm for a matrix.

\section{Existing tests for cross-sectional dependence based on sample correlation} \label{sec:existing}
Consider the heterogeneous panel data model:
\begin{equation} \label{model:hete}
    y_{it}=\x'_{it}\bbeta_i+v_{it}, \quad \text{for} \;  i=1, \cdots, n;\; t= 1, \cdots, T,
\end{equation}
where $i$ indexes the cross-sectional units and $t$ the time series observations, $y_{it}$ is the response variable and $\x_{it}$ is a $k \times 1$ vector of regressors with unity on the first row with coefficients $\bbeta_i$ allowed to vary cross the cross-sectional units. For each $i$, the error term, $v_{it}$, are assumed to
be serially independent with zero mean and finite variance. The null hypothesis of interest in the literature is
\begin{equation*} \label{nullh}
H_0 : v_{it}\; \text{is independent of}\; v_{jt}, \; \text{for all} \; t \; \text{and} \; i\neq j.
\end{equation*}
When $T$ is sufficiently large, a natural way to test $H_0$ is based on some reliable estimates ($\hat{\rho}_{ij}$) for the pair-wise error sample correlations ($\rho_{ij}$). Specifically,
\begin{equation}
    \hat{\rho}_{ij}=\hat{\rho}_{ji}=\frac{\sum_{t=1}^T\hat{v}_{it}\hat{v}_{jt}}{\left(\sum_{t=1}^T\hat{v}_{it}^2\right)^{1/2}\left(\sum_{t=1}^T\hat{v}_{jt}^2\right)^{1/2}},
\end{equation}
where $\hat{v}_{it}$ is the Ordinary Least Squares (OLS) estimate of $v_{it}$ in (\ref{model:hete}) defined by 
\begin{equation} \label{estimate:residual}
    \hat{v}_{it}= y_{it}-\x'_{it}\hat{\bbeta}_i,
\end{equation}
with $\hat{\bbeta}_i$ being the OLS estimate of $\bbeta_i$ by regressing the $T$ sample observations $y_{it}$ on $\x_{it}$ for each $i$. In the seemingly unrelated regression equations (SURE) context with fixed $n$ and as $T\rightarrow\infty$, \cite{breusch1980lagrange} proposed a Lagrange multiplier ($LM$) test for testing $H_0$ given by
\begin{equation} \label{test:lm}
    LM=\frac{T}{2}\sum_{1\leq i\neq j\leq n}\hat{\rho}_{ij}^2=\frac{T}{2}\left[tr(\hat{\R}^2)-n \right],
\end{equation} 
 where $\hat{\R}=(\hat{\rho}_{ij})_{1\leq i, j\leq n}$ is the sample correlation matrix of residuals. It has been shown that $LM\cvd \chi^2_d$, where $d=\frac{1}{2}n(n-1)$  under $H_0$ and normal errors assumption. However, it is well known that the $LM$ test is severely oversized when $n$ is relatively large compared to $T$. To amend this size distortion, \cite{pesaran2004general} put forward a scaled version of the $LM$ test given by
\begin{equation} \label{test:cdlm}
    CD_{LM}=\sqrt{\frac{1}{4n(n-1)}}\sum_{1\leq i\neq j\leq n}(T\hat{\rho}_{ij}^2-1)=\sqrt{\frac{T^2}{4n(n-1)}}\left[tr(\hat{\R}^2)-n-\frac{n(n-1)}{T} \right],
\end{equation}
which is asymptotically distributed as $N(0,1)$ under the SEQ-L scheme. There are two important cases where the $CD_{LM}$ test is not reliable. Firstly, as \cite{baltagi2012lagrange} noted, it will exhibit substantial size distortions in the homogeneous panel data model. Secondly, \cite{pesaran2004general} pointed out that, in finite $T$ case, the $CD_{LM}$ test tends to over-reject the null due to the fact that $E(T\hat{\rho}_{ij}^2-1)$ is not correctly centered at zero. This kind of bias even accumulates as $n$ becomes larger. In this case, \cite{dufour2002exact} suggested to apply bootstrap method to (\ref{test:cdlm}).   \cite{pesaran2004general} and \cite{pesaran2015testing} proposed an alternative adjustment based on the raw, non-squared, sample correlation coefficients given by
\begin{equation} \label{test:cd_p}
    CD_{P}=\sqrt{\frac{T}{2n(n-1)}}\sum_{1\leq i\neq j\leq n}\hat{\rho}_{ij}.
\end{equation}
The test is asymptotically distributed as standard normal under both SEQ-L and SIM-L schemes. However, it is widely reported that
the $CD_P$ test suffers from a specific loss of power when the loadings have zero mean in the cross-sectional dimension under factor representation.  \cite{baltagi2012lagrange} proposed another modified version of the $CD_{LM}$ test for fixed effect panel data model, $LM_{bc}$, given by
\begin{align} \label{test:lmbc}
     LM_{bc}=CD_{LM}-\frac{n}{2(T-1)}=\sqrt{\frac{T^2}{4n(n-1)}}\left[tr(\hat{\R}^2)-n-\frac{n(n-1)}{T}-\frac{n\sqrt{n(n-1)}}{T(T-1)} \right].
\end{align}
The test is asymptotically standard normal with normal errors and strictly exogenous regressors under the SIM-L scheme. \cite{pesaran2008bias} proposed an alternative finite sample adjustment to the $LM$ test by deriving the exact moments of the squared sample correlation coefficients under normal errors and strictly exogenous regressors assumptions. Their $LM_{adj}$ test statistic is given by
\begin{equation} \label{test:lmadj}
    LM_{adj}=\sqrt{\frac{1}{2n(n-1)}}\sum_{1\leq i\neq j\leq n}\frac{(T-k)\hat{\rho}_{ij}^2-\mu_{T,i,j}}{\sigma_{T,i,j}},
\end{equation}
where $$\mu_{T,i,j}=\frac{1}{T-k}tr(\MM_i\MM_j), \; \sigma_{T,i,j}^2=\left[tr(\MM_i\MM_j) \right]^2a_{1T}+ tr(\MM_i\MM_j)^2a_{2T},$$
and $$a_{1T}=a_{2T}-\frac{1}{(T-k)^2}, \; a_{2T}=\frac{3}{(T-k+2)^2},$$
$\MM_i=\mathbf{I}_T-\XX_i'(\XX_i\XX_i')^{-1}\XX_i$ is the projection matrix, where $\XX_i=(\x_{i1}, \cdots, \x_{iT})$ contains $T$ samples on the $k$ regressors for the $i$-th individual regression.  Under $H_0$ and the SEQ-L scheme, $LM_{adj}$ was shown to be asymptotically distributed as $N(0,1).$ However, as pointed out by \cite{pesaran2008bias}, the $LM_{adj}$ test is not robust in panel data models with weakly exogenous regressors. \cite{bailey2021lagrange} proposed another modified $LM$ test for heterogenous panel data models based on Random Matrix Theory, $LM_{RMT}$, given by
\begin{align} \label{test:lmrmt}
LM_{RMT}=\sigma_{RMT}^{-1}\left[{tr(\hat{\R}^2)-n-\frac{n^2}{T}-\frac{n^2}{T^2}+\frac{n}{T}}\right],
\end{align}
where 
\begin{equation} \label{sig:lmrmt}
    \sigma_{RMT}=\frac{4n(2n+T)(n+2T)}{T^3}-\frac{4(\kappa-1)n(n+T)^2}{T^3}-\frac{(\kappa-3)n(n-4T)^2(n+T)^2}{T^5}
\end{equation}
with $\kappa=\frac{3T(T-k-2)}{(T+2)(T-k)}$. Under the assumptions of normal regressors and normal errors, the authors showed that $LM_{RMT}$ is asymptotically distributed as $N(0,1)$ under the SIM-L scheme. The application scopes of the discussed tests are summarized in Table \ref{tab:comparison}.  (The table also contains the new tests proposed in this paper in the last two columns so-called $RLM$ and $RLM_{PE}$, which are developed later.)

\begin{table}
\caption{Application scope of each test for cross-sectional independence   \label{tab:comparison}}
\begin{center}
\begin{tabular}{ccccccc}
 & $CD_P$ & $LM_{adj}$ & $LM_{bc}$ & $LM_{RMT}$& $RLM$ & $RLM_{PE}$\\\hline
Heterogeneous coefficients  & \checkmark  & \checkmark  & $\boldsymbol{?}$ &\checkmark &\checkmark &\checkmark  \\\
Fixed effects panels  & *  & *  & \checkmark & * &\checkmark &\checkmark \\
Dynamic panels  &  \checkmark & * & \checkmark & * &\checkmark & \checkmark\\
Weakly exogenous regressors  &  * & $\boldsymbol{?}$ & $\boldsymbol{?}$ &  *&\checkmark & \checkmark\\
Non-normal errors  &  \checkmark & * & *& * & \checkmark  & \checkmark\\
SIM-L & \checkmark & $\boldsymbol{?}$& \checkmark & \checkmark &\checkmark & \checkmark\\\hline
\end{tabular}
\end{center}
Notes: ``$\checkmark$" signifies that the test can be applied in corresponding panel setting with a theoretical justification.  ``*" means the test is empirically validated without however a theoretical basis.  ``$\boldsymbol{?}$" denotes that it remains unclear in the literature whether such test is applicable. 
\end{table}

\section{The RLM test and its power enhancement} \label{sec:test}
\subsection{The RLM test} \label{sec:rlm}
Motivated by the existing well-known tests based on the sum of squared sample correlation coefficients in (\ref{test:lm}), (\ref{test:cdlm}), (\ref{test:lmbc}), (\ref{test:lmadj}) and (\ref{test:lmrmt}), it is natural to consider the limiting behavior of $tr(\hat{\R}^2)$ under the SIM-L scheme. Throughout the paper, we consider the following assumptions.
\begin{assumption} \label{assum:siml}
$T \rightarrow \infty, n=n(T)\rightarrow \infty$ such  that $c_T=\frac{n}{T}\rightarrow c \in (0,\infty)$.
\end{assumption}

\begin{assumption} \label{assum:disturb}
  For each $i$, the errors, $\{v_{it}\}$, are i.i.d distributed with mean $0$ and variance $\sigma_i^2$. 
\end{assumption}
\begin{assumption}\label{assum:moment}
\begin{enumerate}
    \item[(i)] The errors have uniformly bounded sixth moment, i.e.\\ $\sup_{i,t}E|v_{it}|^{6+\epsilon}\leq C_1$ for some positive constant $C_1$ and $\epsilon>0$.
    \item[(ii)] The errors have uniformly bounded eighth moment, i.e. $\sup_{i,t}E|v_{it}|^{8+\epsilon}\leq C_2$ for some positive constant $C_2$ and $\epsilon>0$.
\end{enumerate}
\end{assumption}
For a static heterogeneous panel data model, we further assume
\begin{assumption} \label{assum:design} For each $i$, the regressors, $\x_{it}$, satisfy
\begin{enumerate}
    \item[(i)] $E(v_{it}| \x_{it}, \cdots, \x_{i1})=0$ for all $i$ and $t$.
    \item[(ii)] let $\XX_i=(\x_{i1},\dots,\x_{iT})$, there exists a  $k\times k$ nonrandom positive definite matrix $\mathbf{B}$ such that $\frac{1}{T}\XX_i\XX_i'\stackrel{p}{\rightarrow}\mathbf{B}$.
     \item[(iii)] $\max_{1\leq t \leq T}\|\frac{1}{\sqrt{T}}\x_{it}\|\stackrel{p}{\rightarrow}0$.
\end{enumerate}
\end{assumption}
Assumption \ref{assum:disturb} is standard allowing for heteroskedastic errors across units. Assumption \ref{assum:moment} requires suitable moments of the errors for the two proposed test procedures, respectively. It helps relax the often-met normal error assumption by Random Matrix Theory. Assumption \ref{assum:design}(i) only requires the regressors to be weakly exogenous. Assumption \ref{assum:design}(ii) and (iii) impose mild conditions on the design matrix. We note that Assumption \ref{assum:design} does not impose the dependence structure between errors and regressors, which allows for the regressors to be weakly exogenous. Under these assumptions, $\sqrt{T}(\hat{\bbeta}_i-\bbeta_i)$ is asymptotically normal according to \cite{lai1982least}.

For a dynamic heterogeneous panel data model with lagged dependent variable included in regressors, more assumptions are needed which will be discussed in Section \ref{sec:dyanamic}.

Now we are in the position of introducing the $RLM$ test and establishing its asymptotic property in the following theorem.

\begin{theorem} \label{trR2:hete}
Under Assumptions \ref{assum:siml}, \ref{assum:disturb}, \ref{assum:moment}$(i)$ and \ref{assum:design}, 
\begin{equation*} \label{test:rlm}
    RLM=\frac{tr(\hat{\R}^2)-\mu_0}{\sigma_0}\cvd N(0,1),
\end{equation*}
where $\mu_0=n+\frac{n^2}{T-1}-c_T$ and $\sigma_0^2=4c_T^2.$
\end{theorem}
The proof of Theorem \ref{trR2:hete} is provided in the Appendix, and the method is in two stages. In the first stage, Lemma 1 establishes the Central Limit Theorem  of $tr(\R^2)$ with tools from Random Matrix Theory, where $\R=(\rho_{ij})_{1\leq i, j \leq n}$. In the second stage, Lemma 2 shows that the asymptotic bias of $tr(\hat{\R}^2)$ disappears under the SIM-L scheme with $c_T \rightarrow c \in (0, \infty)$.

\subsubsection{Relationship between the $RLM$ and $LM_{RMT}$ tests}

By the respective definitions of the $LM_{RMT}$ and $RLM$ in (\ref{test:lmrmt}) and Theorem \ref{test:rlm}, we have
\begin{equation*}
    \mu_0=n(1+\frac{n}{T-1})-c_T=n\left(1+\frac{n}{T}(1+\frac{1}{T-1})\right)-c_T=n+\frac{n^2}{T}+\frac{n^2}{T^2}-\frac{n}{T}+o(1),
\end{equation*}
and
\begin{equation*}
    \sigma_0^2=\sigma_{RMT}^2+o(1)
\end{equation*}
as $\kappa=3+o(1)$. It follows that 
\begin{equation} \label{rlm-lmrmt}
    RLM=LM_{RMT}+o(1),
\end{equation}
which indicates that $LM_{RMT}$ is asymptotically equivalent to $RLM$ regardless of model specifications and assumptions. 
Note that the proof the asymptotic normality of $LM_{RMT}$ in \cite{bailey2021lagrange} 
heavily relies on the assumptions of normal regressors and normal errors as it is used to ensure the residuals have desirable properties, and then transform the sample correlation matrix of residuals to the sample correlation matrix of a nomarlized population with unit covariance matrix (see details in Section 3.1 in \cite{bailey2021lagrange}). From (\ref{rlm-lmrmt}), we conclude that the $LM_{RMT}$ is also valid without the restrictive assumptions of normality. Besides, as we will show later, $RLM$ is also valid in both dynamic and fixed effects panel data models, which theoretically extends the application scope of $LM_{RMT}$. This finding is consistent with the simulation findings that show such robustness of  $LM_{RMT}$ in \cite{bailey2021lagrange}.

\subsubsection{Relationship of the $RLM$, $LM_{bc}$ and $CD_{LM}$ tests}
By the respective definitions of the $CD_{LM}$, $LM_{bc}$ and  $RLM$ tests in (\ref{test:cdlm}), (\ref{test:lmbc}) and Theorem \ref{test:rlm}, we have the following identities
\begin{equation}
    CD_{LM}=LM_{bc}+\frac{n}{2(T-1)},
\end{equation}
\begin{equation}
    CD_{LM}=\sqrt{\frac{n}{n-1}}\left(RLM+\frac{n}{2(T-1)}\right)
\end{equation}
and
\begin{equation} \label{lmbc=rlm}
    LM_{bc}=\sqrt{\frac{n}{n-1}}\left(RLM+\frac{\sqrt{n}}{2(T-1)(\sqrt{n}+\sqrt{n-1})}\right).
\end{equation}
Note that the factor $\sqrt{\frac{n}{n-1}}\rightarrow 1$ and the remainder $\frac{\sqrt{n}}{2(T-1)(\sqrt{n}+\sqrt{n-1})}\rightarrow 0$. It follows that the two tests, $RLM$ and $LM_{bc}$, are always asymptotically equivalent, while the $CD_{LM}$ statistic has always a positive mean shift of value $\frac{n}{2(T-1)}$.

In particular, Theorem \ref{trR2:hete} is also valid for the $LM_{bc}$ statistic. Moreover, anticipating Theorems \ref{trR2:dynamic} and \ref{trR2:fixed} in Sections \ref{sec:dyanamic} and \ref{sec:fixed}, for dynamic and fixed effects panel data model, respectively, these asymptotic normality are also valid for the $LM_{bc}$ statistic. In this sense, the results from the paper can also be considered as new extension of the $LM_{bc}$ test, originally developed for homogeneous fixed effects panel data model in \cite{baltagi2012lagrange}, to various large panel models with coefficient heterogeneity.

\subsection{The $RLM_{PE}$ test}
In the high dimensional setting, for testing the identity hypothesis $H_0: \mathrm{corr}(\mathbf{v}_t)=\mathbf{I}_n$, where $\mathbf{v}_t=(v_{1t},\dots,v_{nt})'$, there are mainly two types of test statistics. The majority of existing tests are based on the squared Frobenius norm $\|\R-\mathbf{I}_n\|_F^2=tr(
\R^2)-n=\sum_{i\neq j}\rho_{ij}^2.$ However, this quadratic statistic lacks power if $ \mathrm{corr}(\mathbf{v}_t)$ is a sparse matrix, see \cite{fan2015power}. Considering this, tests based on the maxima of absolute values, $\max_{i<j}|\rho_{ij}|$, which share a asymptotic type I extreme value distribution, are generally powerful under sparse alternatives. This approach has however a main drawback that such test can suffer from size distortions, which is common for statistics of the maximum type, see \cite{liu2008asymptotic}. Besides, this way is not as appropriate as the Frobenius norm (sum) type in some cases. For example, consider the alternative
\begin{equation*}
    H_1:  \mathrm{corr}(\mathbf{v}_t)=\mathbf{I}_n+\mathbf{P}_n,
\end{equation*}
where $\mathbf{P}_n$ is a perturbation matrix with diagonal entries being zero and  $s$ non-zero off diagonal entries, where $s \in \{1, \dots, n^2-n \}$. Intuitively, $\mathbf{P}_n$ can be designed as a dense matrix but with weak coefficients such that $\max_{i<j}|\rho_{ij}|<z_\alpha$ for any $s$. Consequently,  the extreme value type tests will fail to detect such a matrix. In such instances, the sum type tests are more suitable in the light of the fact that the eigenvalues of $\R$ could vary from $H_0$ to $H_1$,  which results in larger quadratic statistic value by $tr(
\R^2)=\sum_{i=1}^n\lambda_i^2({\R})$.  In order to realize an interpolation of the two types of statistics above, namely the maximum type and the sum type, we propose a new test statistic based on $tr(\R^4)=\sum_{i=1}^n\lambda_i^4({\R})$. The reason is that large empirical correlations, $\rho_{ij}^4$, would be more emphasized in $tr(\R^4)$ than in $tr(\R^2)$. To see this, consider increasingly large powers of the sample correlations, $\rho_{ij}^m$, where $m$ is a positive integer.  Let $E=\text{argmax}_{(i,j):\;i<j}|\rho_{ij}|$, then \begin{equation} \label{eq:equiv}
     \sum_{i\neq j}|\rho_{ij}|^m \sim \text{card}(E)\cdot (\max_{i<j}|\rho_{ij}|)^m, \; \text{as} \;m\rightarrow \infty,
 \end{equation}
where $\mathrm{card}(E)$ denotes the cardinality of the set $E$. Therefore, the new statistic with $m=4$ can mimic some properties of the maximum type, while remaining a sum type smoothing statistic. The resulting power is expected to be higher than $tr(\R^2)$ when very few sample correlations are significantly non-zero under sparse alternatives, and higher than maximum type statistics when there are many but relatively small correlations.

\subsubsection{Test based on $\sum_{i\neq j}\hat{\rho}_{ij}^4$}
On the ground of analyses above, we propose a new test statistic based on the fourth power of $\hat{\rho}_{ij}$ in the following theorem.
\begin{theorem} \label{trR4:hete}
Under Assumptions \ref{assum:siml}, \ref{assum:disturb}, \ref{assum:moment}$(ii)$ and \ref{assum:design}, 
\begin{equation*} \label{test:rlmpe}
    RLM_{PE}=\frac{tr(\hat{\R}^4)-\mu_{PE}}{\sigma_{PE}}\cvd N(0,1)
\end{equation*}
where $\mu_{PE}=n+\frac{6n^2}{T-1}+\frac{6n^3}{(T-1)^2}+\frac{n^4}{(T-1)^2}-6c_T(1+c_T)^2-2c_T^2$ and $\sigma_{PE}^2=8c_T^2+96c_T^3(1+c_T)^2+16c_T^2(3c_T^2+8c_T+3)^2.$
\end{theorem}

\begin{remark}
We choose $m=4$ to generate the $RLM_{PE}$ test for technical simplicity. In fact, one can increase $m$ to any large even integer to obtain new tests that may have larger power in the sparse correlation setting. This strategy is feasible with the proof techniques provided in Appendix. Further, by (\ref{eq:equiv}), it is expected that tests based on $\sum_{i\neq j}|\hat{\rho}_{ij}|^m$ would share similar power with maximum type statistics, which has been suggested as a powerful tests in sparse data, for example, see \cite{cai2014two}.  It can provide a series of potential statistics that can well control the size and may be more powerful under sparse alternatives at the same time.
\end{remark}
\begin{remark}
For the initial $RLM$ test (also the $LM_{bc}$ test), the remarkable screening technique in \cite{fan2015power} can provide an improved test that has the same asymptotic size with non-inferior asymptotic power against a broader range of alternatives. Compared to this approach, our power enhanced test avoids constructing such a ``power enhancement component" by increasing the power of the sample correlation to four. However, studying the power properties of our technique with different choices of $m$ is not the main focus in this paper and remains an open problem. 
\end{remark}
The proof of the Theorem \ref{trR4:hete} is similar to that of Theorem \ref{trR2:hete}, which requires the two lemmas given in Appendix.

\section{Dynamic panel data model} \label{sec:dyanamic}
In this section, we show that the $RLM$ and $RLM_{PE}$ tests are asymptotically valid in a dynamic panel data model, which is specified as following:
\begin{equation} \label{model:dynamic}
    y_{it}=\alpha_i y_{i,t-1}+\x_{it}'\bbeta_i+v_{it},
\end{equation}
for $i=1,\dots,n; t=1,\dots,T$, where $y_{i,t-1}$ is the lagged dependent variable. Let  $\z_{it}=(y_{i,t-1}, \x_{it}')$, $\pphi_i=(\bbeta_i,\alpha_i)'$, then (\ref{model:dynamic}) can be rewritten as $y_{it}=\z_{it}'\pphi_i+v_{it}$. We show that the proposed $RLM$ and $RLM_{PE}$ tests still have standard normal limiting distribution under the null hypothesis in the dynamic panel data model. To establish the asymptotic normality, we need additional assumptions as following,
\begin{assumption} \label{assum:se}
\begin{enumerate}
    \item[(i)] $\{y_{it}\}_{1\leq i \leq n, 1\leq t \leq T}$ is a stationary and ergodic process.
     \item[(ii)]  Let $\Y_i=(y_{i,0},\dots,y_{i,T-1})$, $ \frac{1}{T}\Y_i\Y_i'=O_p(1)$ holds uniformly in $i$.
\end{enumerate}
\end{assumption}
We establish the limiting distributions of the proposed tests in the following theorems
\begin{theorem} \label{trR2:dynamic}
 Under Assumptions \ref{assum:siml}, \ref{assum:disturb}, \ref{assum:moment}$(i)$, \ref{assum:design} and \ref{assum:se},
$$RLM\cvd N(0,1).$$
\end{theorem}

\begin{theorem} \label{trR4:dynamic}
 Under Assumptions \ref{assum:siml}, \ref{assum:disturb}, \ref{assum:moment}$(ii)$, \ref{assum:design} and \ref{assum:se},
$$RLM_{PE}\cvd N(0,1).$$
\end{theorem}
Under Assumption \ref{assum:se}, the proofs of Theorems \ref{trR2:dynamic} and \ref{trR4:dynamic} follow along the same lines as that of static panel data model. See the Appendix.

\section{Fixed effects panel data model} \label{sec:fixed}
In this section, we establish the asymptotic normality of the $RLM$ and $RLM_{PE}$ tests in a fixed effects panel data model. We find that as long as the coefficient estimator is $\sqrt{T}$-consistent, the proposed tests still have standard normal limiting distribution under the null. To allow for weakly exogenous regressors, various consistent estimators have been proposed in the literature including \cite{chudik2015common}, \cite{chudik2018half} etc. However, these estimators require stronger assumptions than the static panel data
model. For simplicity of illustration, we focus on residuals obtained by the within estimator. The strictly exogenous assumption is then necessary for the consistency of the within estimator. One can relax this assumption to the weakly exogenous one by applying a $\sqrt{T}$-consistent estimator.

Consider a fixed effects panel data model:
\begin{equation} \label{model:fixed}
    y_{it}=\x_{it}'\bbeta + \mu_i +v_{it},
\end{equation}
for $i=1,\dots,n; t=1,\dots,T$, where $\mu_i$ denotes the time-invariant individual effect. The within estimator in (\ref{model:fixed}) is specified by 
\begin{equation} 
\hat{\bbeta}=\left(\sum_{t=1}^T\sum_{i=1}^n\tilde{\x}_{it}\tilde{\x}_{it}'\right)^{-1}\left(\sum_{t=1}^T\sum_{i=1}^n\tilde{\x}_{it}\tilde{y}_{it}\right),
\end{equation}
where $\tilde{\x}_{it}=\x_{it}-\frac{1}{T}\sum_{t=1}^T\x_{it}$ and $\tilde{y}_{it}=y_{it}-\frac{1}{T}\sum_{t=1}^Ty_{it}$. 
\begin{assumption} \label{assum:design2} The regressors, $\x_{it}$, satisfy
  \begin{enumerate}
      \item[(i)](strictly exogenous) $E(v_{it}| \x_{iT}, \cdots, \x_{i1})=0$  and  $E(v_{jt}| \x_{iT}, \cdots, \x_{i1})=0$ for all $i,j$ and $t$.
      \item[(ii)] For the demeaned regressors $\tilde{\x}_{it}$, $\frac{1}{T}\sum_{t=1}^T\tilde{\x}_{it}$ and $\frac{1}{T}\sum_{t=1}^T\tilde{\x}_{it}\tilde{\x}_{jt}'$ are stochastic bounded for all $i,j$. Besides, $\lim_{(n,T)\rightarrow \infty}\frac{1}{nT}\sum_{i=1}^n\sum_{t=1}^T\tilde{\x}_{it}\tilde{\x}_{it}'$ exists and is nonsingular.
  \end{enumerate}
\end{assumption}
Under the Assumptions \ref{assum:siml}, \ref{assum:disturb}, \ref{assum:moment} and \ref{assum:design2}, $\hat{\bbeta}$ is $\sqrt{nT}-$consistent. We establish the validity of our proposed tests in the following theorems.

\begin{theorem} \label{trR2:fixed}
Under Assumptions \ref{assum:siml}, \ref{assum:disturb}, \ref{assum:moment}$(i)$ and \ref{assum:design2},
$$RLM\cvd N(0,1).$$
\end{theorem}
\begin{theorem} \label{trR4:fixed}
Under Assumptions \ref{assum:siml}, \ref{assum:disturb}, \ref{assum:moment}$(i)$ and \ref{assum:design2},
$$RLM_{PE}\cvd N(0,1).$$
\end{theorem}
The proofs of Theorems \ref{trR2:fixed} and \ref{trR4:fixed} are given in the Appendix.

\section{Monte Carlo simulations} \label{sec:simulation}
In this section, we conduct Monte Carlo simulations to examine the empirical sizes and powers of our $RLM$ and $RLM_{PE}$ tests, which are defined by (\ref{test:rlm}) and (\ref{test:rlmpe}), respectively, and compare their performances to that of the $CD_p$ test and the $LM_{adj}$ test defined by (\ref{test:cd_p}) and (\ref{test:lmadj}), respectively. We consider four data generating processes (DGPs): heterogeneous panel data model with either strictly or weakly exogenous regressors, fixed effects panel data model and pure dynamic panel data model.

Before looking at the simulation results, we consider the estimated rejection frequencies within range from 3.6\% to 6.5\% to provide evidence consistent with the robustness of the tests, following the arguments in \cite{halunga2017heteroskedasticity}. Besides, we don't include the $LM_{RMT}$ and the $LM_{bc}$ tests since they are almost identical to the $RLM$ test by (\ref{rlm-lmrmt})  and (\ref{lmbc=rlm}). 

\subsection{Monte Carlo design}
\subsubsection{DGP1: Heterogeneous panel data model with strictly exogenous regressors}
We first consider the DGP used in \cite{pesaran2008bias}, which is specified by
\begin{equation*}
    y_{it}=\alpha_i+\sum_{l=2}^kx_{lit}\beta_{li}+v_{it}, \;\;\;\; i=1,2,\dots, n; \;\;\; t=1,2,\dots, T,
\end{equation*}
where $\alpha_i\sim IIDN(1,1)$, $\beta_{li}\sim IIDN(1,0.04)$. The regressors are generated as
\begin{equation*}
    x_{lit}=0.6x_{lit-1}+u_{lit}, \;\;\;\; i=1,2,\dots, n; \;\;\; t=-50,\dots,0,\dots, T; \; l=2,\dots, k
\end{equation*}
with $x_{li,-51}=0$ where  $u_{lit}\sim IIDN(0, \tau_{li}^2/(1-0.6^2))$, $\tau_{li}^2\sim IID \chi^2(6)/6$.  The first 50 observations are discarded to lessen the effects of initial values. Now we generate the disturbances under the null $H_0$ as $v_{it}=\sigma_i\epsilon_{it}$, where $\sigma_i\sim\chi^2(2)/2 $ and $\{\epsilon_{it}\}$ are generated from three different distributions: (i) normal, $N(0,1)$, (ii) chi-squared, $(\chi^2(5)-5)/\sqrt{10}$ and (iii) student-t, $t_{10}/\sqrt{10/8}$. The normalizations in (ii) and (iii) are such that errors have mean one and variance one. To investigate the effects of  the number of regressors, $k=2, 4$ are considered. 

To examine the powers of the proposed tests, the disturbances are generated by a factor model as following:
\begin{equation}
    v_{it}=\lambda_if_t+\epsilon_{it},
\end{equation}
where $f_t (t=1,\dots,T)$ are the factors with $f_t\sim IIDN(0,1)$ and $\lambda_i (i=1,\dots,n)$ are the loadings. We consider the following three cases of loading construction:

\begin{enumerate}
    \item[(1)] Dense case. $\lambda_i\sim IIDU(-b,b)$, for $i=1,\dots,n$, where $b=\sqrt{3h/n}$ and $h=3.$
    \item[(2)] Sparse case.  $\lambda_i\sim IIDU(0.5,1.5)$, for $i=1,\dots,[n^{0.3}]$, and $\lambda_i=0$, for $i=[n^{0.3}]+1,\dots,n$, where $[n^{0.3}]$ is the integer part of $n^{0.3}$.
    \item[(3)] Less-sparse case.  $\lambda_i\sim IIDU(0.5,1.5)$, for $i=1,\dots,[n^{0.5}]$, and $\lambda_i=0$, for $i=[n^{0.5}]+1,\dots,n$.
\end{enumerate}
In the dense case, $h$ measures the degree of cross-sectional dependence. The sparse case and the less-sparse case follow the design used in \cite{bailey2016exponent} to model the weak and strong cross-sectional dependence, respectively. The Monte Carlo experiments are conducted for $T=50, 100 ,200$, and three different choices of ratio $n/T=0.5, 1, 2$ basing on 2000 replications. To obtain the empirical size, the proposed $RLM$ test, $RLM_{PE}$ test and $LM_{adj}$ test are implemented at the one-sided 5\% nominal significance level, while $CD_P$ test is conducted at the two-sided 5\% nominal significance level.

\setlength{\tabcolsep}{1mm}{
\begin{table}[]
\renewcommand\arraystretch{0.8}
\scriptsize
    \centering
    
    \begin{tabular}{ccccccccccc}
    \hline
   \hline
   \multicolumn{11}{c}{$k=2$}\\
   \hline 
   & & \multicolumn{3}{c}{Chi-squared} & \multicolumn{3}{c}{Normal} & \multicolumn{3}{c}{Student-t}  \\
\cmidrule(r){3-5} \cmidrule(r){6-8} \cmidrule(r){9-11} 
&$(T,n)$ & (50,25) & (100,50) & (200,100)  
& (50,25) & (100,50) & (200,100)   
& (50,25) & (100,50) & (200,100)  \\  
\midrule
     \multirow{4}{*}{$\frac{n}{T}=\frac{1}{2}$ }    &$RLM$& 5.15  & 5.15  & 5.55 & 5.55   &4.65  & 5.35 & 4.9 & 5.25  &5.15    \\
         &$RLM_{PE}$&  5.4 &  5.7 & 5.6 &   5.2 &4.45  &5.8  & 4.8 &  5 & 5.55 \\
         &$LM_{adj}$& 5.85  & 5.35  & 5.85  &  6  & 4.7 & 5.45 &5.4  &  5.45 & 5.25 \\
         &$CD_{P}$&  4.75 & 4.45 & 4.65 &  4.85  &4.25  &5.25  &4.5  & 4.15  &  5.1 \\
\midrule
   & & \multicolumn{3}{c}{Chi-squared} & \multicolumn{3}{c}{Normal} & \multicolumn{3}{c}{Student-t}    \\
\cmidrule(r){3-5} \cmidrule(r){6-8} \cmidrule(r){9-11}  
&$(T,n)$ & (50,50) & (100,100) & (200,200)  
& (50,50) & (100,100) & (200,200)    
& (50,50) & (100,100) & (200,200)    \\  
\midrule
     \multirow{4}{*}{$\frac{n}{T}=1$ }    &$RLM$&  5.5 & 5.15 & 4.55  &  5.05  & 4.5 & 4.95 & 5.15 & 4.95  &  6  \\
         &$RLM_{PE}$ &  5.5 & 5.9 & 4.75 &  4.45  & 4.85 & 4.9 & 5.05 & 4.7  &5.6  \\
         &$LM_{adj}$& 5.7  & 5.25 & 4.6 &  5.1  & 4.6 & 5 & 5.35&  4.95 & 6   \\
         &$CD_{P}$ &  5.4 & 4.85 & 5.05  &  5.25  & 5.2 & 5.1 & 5.2&  4.75 &  5.2  \\
         
\midrule
  & & \multicolumn{3}{c}{Chi-squared} & \multicolumn{3}{c}{Normal} & \multicolumn{3}{c}{Student-t}  \\
\cmidrule(r){3-5} \cmidrule(r){6-8} \cmidrule(r){9-11} 
&$(T,n)$ & (50,100) & (100,200) & (200,400)  
& (50,100) & (100,200) & (200,400)    
& (50,100) & (100,200) & (200,400)     \\  
\midrule
     \multirow{4}{*}{$\frac{n}{T}=2$ }    &$RLM$&  5.8 & 5.3 &6.2  &  5.1  &5.45  &5  &5.4  & 5.6  & 5.5  \\
         &$RLM_{PE}$ &  5.5 & 5 & 5.55 &   4.95 &5.55  &4.65  & 5.4 & 5.85  & 5.35\\
         &$LM_{adj}$& 5.7  & 5.3  & 6.2 &  5.05  &5.45  &5  & 5.2 &5.5   & 5.5 \\
         &$CD_{P}$ & 4.05  & 4.75 & 5.7 &  5  & 4.85 & 5.4 & 4.85 & 5.05  &4.9 \\
         
          \hline
   \hline
     \multicolumn{11}{c}{$k=4$}\\
   \hline 
    & & \multicolumn{3}{c}{Chi-squared} & \multicolumn{3}{c}{Normal} & \multicolumn{3}{c}{Student-t} \\
\cmidrule(r){3-5} \cmidrule(r){6-8} \cmidrule(r){9-11} 
&$(T,n)$ & (50,25) & (100,50) & (200,100)  
& (50,25) & (100,50) & (200,100)   
& (50,25) & (100,50) & (200,100)  \\  
\midrule
     \multirow{4}{*}{$\frac{n}{T}=\frac{1}{2}$ }    &$RLM$&  5.65 & 5.05 & 5.65 & 5.5   &5.25  & 5.4 & 5.3 & 5.4  &5.3     \\
         &$RLM_{PE}$&  5.4 & 5.3 &  5.5     &   5.35 & 5.2 &5  & 5.4 &  5.65 & 5.45 \\
         &$LM_{adj}$&  5.75 & 5.05    &  5.75   &   5.6 &5.25  &  5.4&5.4  & 5.45  &  5.3 \\
         &$CD_{P}$& 4.95  & 5.1  &  4.55 &   5.45 & 5.3 &5.1  & 4.85 &  5.35 &  5.2  \\
\midrule
  & & \multicolumn{3}{c}{Chi-squared} & \multicolumn{3}{c}{Normal} & \multicolumn{3}{c}{Student-t} \\
\cmidrule(r){3-5} \cmidrule(r){6-8} \cmidrule(r){9-11}  
&$(T,n)$ & (50,50) & (100,100) & (200,200)  
& (50,50) & (100,100) & (200,200)    
& (50,50) & (100,100) & (200,200)    \\  
\midrule
     \multirow{4}{*}{$\frac{n}{T}=1$ }    &$RLM$&  7.3 & 5 & 5.25 &  6.45  &6  &5.35  &6.4  & 6.25 & 5.3  \\
         &$RLM_{PE}$&  7.1& 5.35 &  5.05 &   5.75 & 5.7 & 5.35 & 6.2& 5.5  &5.1  \\
         &$LM_{adj}$& 6.7  & 4.55 &  5.2 &  5.75  &5.6  & 5.2 & 5.6 &   5.7& 5.1   \\
         &$CD_{P}$ & 5.55  & 4.75 & 5 &  4.35  &4.7  &5.25  & 4.3 & 4.4  &  5.35  \\
         
\midrule
  & & \multicolumn{3}{c}{Chi-squared} & \multicolumn{3}{c}{Normal} & \multicolumn{3}{c}{Student-t}  \\
\cmidrule(r){3-5} \cmidrule(r){6-8} \cmidrule(r){9-11} 
&$(T,n)$ & (50,100) & (100,200) & (200,400)  
& (50,100) & (100,200) & (200,400)    
& (50,100) & (100,200) & (200,400)     \\  
\midrule
     \multirow{4}{*}{$\frac{n}{T}=2$ }    &$RLM$& 6.75  & 5.85 &  6.25 &  7.65  &5.9  &5.55  & 7.55 &  5.75 & 4.95  \\
         &$RLM_{PE}$ &  6.25 & 6.15& 5.55 &  7.4  & 5.75 & 5.55 &7.25  &5.75   & 4.95 \\
         &$LM_{adj}$&  5 &5  & 5.55 &  5.8  & 4.75 &  5.1& 5.65 & 4.9  &    4.65\\
         &$CD_{P}$& 5.9  & 5.7  & 4.85 & 5.15   & 5.3 & 5.5 &4.45  & 5.5  & 4.95\\
\hline\hline

    \end{tabular}
    \caption{Empirical size of tests in DGP1}
    \label{tab:DGP1}
    
\end{table}}

\setlength{\tabcolsep}{1mm}{
\begin{table}[]
\renewcommand\arraystretch{0.8}
\scriptsize
    \centering
    \begin{tabular}{ccccccccccc}
    \hline
   \hline
   \multicolumn{11}{c}{$k=2,h=3$}\\
   \hline 
   & & \multicolumn{3}{c}{Chi-squared} & \multicolumn{3}{c}{Normal} & \multicolumn{3}{c}{Student-t}  \\
\cmidrule(r){3-5} \cmidrule(r){6-8} \cmidrule(r){9-11} 
&$(T,n)$ & (50,25) & (100,50) & (200,100)  
& (50,25) & (100,50) & (200,100)   
& (50,25) & (100,50) & (200,100)  \\  
\midrule
     \multirow{4}{*}{$\frac{n}{T}=\frac{1}{2}$ }    &$RLM$& 91.45  & 98.9   & 99.85 & 99.55   &98.75  & 100  & 99.6 & 99.75  & 99.8    \\
         &$RLM_{PE}$ & 95.95  & 99.85  & 100  & 99.85 & 99.75  & 100  & 99.85  & 100   & 99.95 \\
         &$LM_{adj}$&92.3   & 98.95  & 99.85  & 99.55   & 98.8 & 100 & 99.6  & 99.8  & 99.8 \\
         &$CD_{P}$& 5.15  & 4.9  & 4.45  &5.29  & 5.05     & 5.1  & 5.35  & 4.75  & 4.7   \\
\midrule
   & & \multicolumn{3}{c}{Chi-squared} & \multicolumn{3}{c}{Normal} & \multicolumn{3}{c}{Student-t}   \\
\cmidrule(r){3-5} \cmidrule(r){6-8} \cmidrule(r){9-11}  
&$(T,n)$ & (50,50) & (100,100) & (200,200)  
& (50,50) & (100,100) & (200,200)    
& (50,50) & (100,100) & (200,200)    \\  
\midrule
     \multirow{4}{*}{$\frac{n}{T}=1$ }    &$RLM$&75.45  &91.45  &  96.7 & 81.25  & 80.55  & 95.45  & 79.95 & 86.1  &   98.1  \\
         &$RLM_{PE}$ & 84.2   & 97.4   &  99.9 & 89.6  & 93.1   &  99.85  & 89.65 & 95.2  & 100  \\
         &$LM_{adj}$& 75.95   & 91.6   &  96.8  & 81.85 &80.8   &    95.45  &80.4 & 86.25  &   98.1 \\
         &$CD_{P}$ & 5.05 & 11  &  5.2  & 4.9  & 4.95  &4.85   & 4.3   & 4.7  &4.35   \\
         
\midrule
  & & \multicolumn{3}{c}{Chi-squared} & \multicolumn{3}{c}{Normal} & \multicolumn{3}{c}{Student-t}  \\
\cmidrule(r){3-5} \cmidrule(r){6-8} \cmidrule(r){9-11} 
&$(T,n)$ & (50,100) & (100,200) & (200,400)  
& (50,100) & (100,200) & (200,400)    
& (50,100) & (100,200) & (200,400)     \\  
\midrule
     \multirow{4}{*}{$\frac{n}{T}=2$ }    &$RLM$& 60.4  & 59.8   & 74.75 & 43   & 60.25  &  72.9  & 62.7  &  55.55  &   62.6  \\
         &$RLM_{PE}$ & 73.85   & 78.55  &  91.35 & 55.35  &78.35  &  91.75 & 75.9   & 73.85   & 84.7\\
         &$LM_{adj}$& 60.2 & 59.8  &74.75   & 42.8  &60.2  &  72.9 & 62.6 & 55.5  & 62.45  \\
         &$CD_{P}$ & 7.2   & 4.75  &  4.7    & 4.65  & 4.6   & 5.4  & 4.9   &5.6  &5.05  \\
         
          \hline
   \hline
  \multicolumn{11}{c}{$k=4,h=3$}\\
   \hline 
    & & \multicolumn{3}{c}{Chi-squared} & \multicolumn{3}{c}{Normal} & \multicolumn{3}{c}{Student-t}  \\
\cmidrule(r){3-5} \cmidrule(r){6-8} \cmidrule(r){9-11} 
&$(T,n)$ & (50,25) & (100,50) & (200,100)  
& (50,25) & (100,50) & (200,100)   
& (50,25) & (100,50) & (200,100)  \\  
\midrule
     \multirow{4}{*}{$\frac{n}{T}=\frac{1}{2}$ }    &$RLM$&90.1  & 99.2   &99.55  & 98.15 & 98.9 & 99.95   & 95.1  & 99.4 & 99.9     \\
         &$RLM_{PE}$ & 94.9   & 99.85  & 100  & 99.4 & 99.85 & 100  &97.55   & 99.85  & 100  \\
         &$LM_{adj}$& 90.2  & 99.2  & 99.55  & 98.15  & 98.95& 99.95 & 95.15  & 99.4   & 99.9  \\
         &$CD_{P}$& 5.2  & 7.65   & 4.4  & 31.65  & 4.85  & 4.85  & 4.7  & 5.55  & 4.8   \\
\midrule
  & & \multicolumn{3}{c}{Chi-squared} & \multicolumn{3}{c}{Normal} & \multicolumn{3}{c}{Student-t}  \\
\cmidrule(r){3-5} \cmidrule(r){6-8} \cmidrule(r){9-11}  
&$(T,n)$ & (50,50) & (100,100) & (200,200)  
& (50,50) & (100,100) & (200,200)    
& (50,50) & (100,100) & (200,200)    \\  
\midrule
     \multirow{4}{*}{$\frac{n}{T}=1$ }    &$RLM$&80.7  & 91.9  & 96.25   & 79.2   & 87.3 & 95.65 & 78.3 &91.85  &   96.55  \\
         &$RLM_{PE}$ & 88.45   & 98.1   & 99.65  & 88.85     & 95.85   &99.55  & 87.6  & 97.95 & 99.65 \\
         &$LM_{adj}$& 79.55 & 91.6  & 96.2 &78.5  & 87  &  95.5  & 76.8 & 91.4  &  96.55  \\
         &$CD_{P}$ & 4.1  & 5.2  &  4.6  & 4.05   & 5 &  4.35 & 4  & 5.2 & 5.35   \\
         
\midrule
  & & \multicolumn{3}{c}{Chi-squared} & \multicolumn{3}{c}{Normal} & \multicolumn{3}{c}{Student-t}  \\
\cmidrule(r){3-5} \cmidrule(r){6-8} \cmidrule(r){9-11} 
&$(T,n)$ & (50,100) & (100,200) & (200,400)  
& (50,100) & (100,200) & (200,400)    
& (50,100) & (100,200) & (200,400)     \\  
\midrule
     \multirow{4}{*}{$\frac{n}{T}=2$ }    &$RLM$&48.15  & 70.95   &    70.15& 49.15  & 60.8  &   68.6 & 47.4   & 60.1 & 65.7 \\
         &$RLM_{PE}$ & 60.1  & 87.5  &  91.05  & 60.75  & 78.6 & 89.1  &58.65 & 78.25  &  86.3 \\
         &$LM_{adj}$&43.35  & 68.3  &  68.75  & 43.65  & 58.25   &  67.1 & 42.65  &56.75 &63.85\\
         &$CD_{P}$ &5.25   &4.7  &4.9     & 4.55  & 5.05   &  4.95  & 7.55  & 5.3  & 4.3\\
\hline\hline

    \end{tabular}
    \caption{Empirical power of of tests in DGP1 for dense case}
    \label{tab:power_dgp1_dense2}
\end{table}}

\setlength{\tabcolsep}{1mm}{
\begin{table}[]
\renewcommand\arraystretch{0.8}
\scriptsize
    \centering
    \begin{tabular}{ccccccccccc}
    \hline
   \hline
   \multicolumn{11}{c}{$k=2$}\\
   \hline 
     & & \multicolumn{3}{c}{Chi-squared} & \multicolumn{3}{c}{Normal} & \multicolumn{3}{c}{Student-t}   \\
\cmidrule(r){3-5} \cmidrule(r){6-8} \cmidrule(r){9-11} 
&$(T,n)$ & (50,25) & (100,50) & (200,100)  
& (50,25) & (100,50) & (200,100)   
& (50,25) & (100,50) & (200,100)  \\  
\midrule
     \multirow{4}{*}{$\frac{n}{T}=\frac{1}{2}$ }    &$RLM$& 20.4  & 28.85  & 43.95 &  20.55 & 25.9  & 39.85 & 20.25 & 20.3  &   27.7  \\
         &$RLM_{PE}$ & 18.35   & 30.65  & 52.25  & 17.2  & 28.15  & 45.6  &  18.75 & 21.05   & 31.25  \\
         &$LM_{adj}$& 21.5  & 29.25 &44.45   & 21.55 & 26.75  & 40.25  &  21.6 & 21  &28   \\
         &$CD_{P}$& 7.45 & 6.75   & 6.65   & 7.15 &  7.8 &  6.9 & 7.4   & 6.85  &  6.5 \\
\midrule
  & & \multicolumn{3}{c}{Chi-squared} & \multicolumn{3}{c}{Normal} & \multicolumn{3}{c}{Student-t}   \\
\cmidrule(r){3-5} \cmidrule(r){6-8} \cmidrule(r){9-11}  
&$(T,n)$ & (50,50) & (100,100) & (200,200)  
& (50,50) & (100,100) & (200,200)    
& (50,50) & (100,100) & (200,200)    \\  
\midrule
     \multirow{4}{*}{$\frac{n}{T}=1$ }    &$RLM$& 23.6  & 13.6  & 38.35  & 14.05 &  18.95 & 39.55  & 11.9  & 23.75  & 31.3 \\
         &$RLM_{PE}$ & 25.6   & 13.8 & 49.3 & 14.15 &  20.5  & 50.5  &  11.95 &   25.85 &  38 \\
         &$LM_{adj}$& 24.2   &13.65   & 38.45   & 14.75 & 19.2  &   39.7 &  12.45 &  24.05  & 31.4 \\
         &$CD_{P}$ & 7.9 & 6  & 7.35  &  7.35 &6   &  7.25 & 6.8  & 5.9  & 6.4 \\
         
\midrule
  & & \multicolumn{3}{c}{Chi-squared} & \multicolumn{3}{c}{Normal} & \multicolumn{3}{c}{Student-t} \\
\cmidrule(r){3-5} \cmidrule(r){6-8} \cmidrule(r){9-11} 
&$(T,n)$ & (50,100) & (100,200) & (200,400)  
& (50,100) & (100,200) & (200,400)    
& (50,100) & (100,200) & (200,400)     \\  
\midrule
     \multirow{4}{*}{$\frac{n}{T}=2$ }    &$RLM$& 12.2   & 11.85 & 27.9   & 10.45 &  18.4 & 38.45  & 11.8  & 19.75  &54.15 \\
         &$RLM_{PE}$ & 13.4   & 12.5  & 35.8   & 10.8 & 20.35  &  52.75 & 12.65  & 21.9  & 72.15\\
         &$LM_{adj}$& 12.1  & 11.8   & 27.8   &10.4  & 18.35  & 38.3  & 11.7  &  19.7 & 54.1   \\
         &$CD_{P}$ & 6.55  & 5.45 & 6.55    & 6.65 &  5.9 &  7.1 & 6.2  & 6  &6.9 \\
         
          \hline
   \hline
  \multicolumn{11}{c}{$k=4$}\\
   \hline 
    & & \multicolumn{3}{c}{Chi-squared} & \multicolumn{3}{c}{Normal} & \multicolumn{3}{c}{Student-t}   \\
\cmidrule(r){3-5} \cmidrule(r){6-8} \cmidrule(r){9-11} 
&$(T,n)$ & (50,25) & (100,50) & (200,100)  
& (50,25) & (100,50) & (200,100)   
& (50,25) & (100,50) & (200,100)  \\  
\midrule
     \multirow{4}{*}{$\frac{n}{T}=\frac{1}{2}$ }    &$RLM$& 10.45 & 24.3   &26.4  & 9.15 &  27.25 & 27  & 17.5  & 51.7  &  28.5   \\
         &$RLM_{PE}$ &10   & 27.85  & 30.95  &  9.65 & 28.65  & 28.9  &  15.85 &  58.05 & 30.35 \\
         &$LM_{adj}$& 10.5   & 24.35  & 26.4  & 9.3 & 27.6  & 27.05  &  17.55 &51.95   &28.5  \\
         &$CD_{P}$& 6.25  &7.75 & 5.95  & 7.1 & 7.75  &7   & 7.1  &9.1   &6.75  \\
\midrule
   & & \multicolumn{3}{c}{Chi-squared} & \multicolumn{3}{c}{Normal} & \multicolumn{3}{c}{Student-t}  \\
\cmidrule(r){3-5} \cmidrule(r){6-8} \cmidrule(r){9-11}  
&$(T,n)$ & (50,50) & (100,100) & (200,200)  
& (50,50) & (100,100) & (200,200)    
& (50,50) & (100,100) & (200,200)    \\  
\midrule
     \multirow{4}{*}{$\frac{n}{T}=1$ }    &$RLM$& 16.5  & 29.9   & 16.45    & 15 & 13.35  &36.35   &  16.85 & 9.5  & 50.3 \\
         &$RLM_{PE}$ & 16.45  & 34   & 18.8  & 15.05 &  13.55 &  46.5 &  17.15 &  10.35 &64.95 \\
         &$LM_{adj}$& 15.35   &28.75   & 16.1   & 13.65 & 12.9  &   35.95& 15.6  &  9.05 &  49.9 \\
         &$CD_{P}$ & 7  & 6.95   &6.2   & 7.4 &  7 & 6.9   & 7.7  &  5.75 &  7\\
         
\midrule
  & & \multicolumn{3}{c}{Chi-squared} & \multicolumn{3}{c}{Normal} & \multicolumn{3}{c}{Student-t} \\
\cmidrule(r){3-5} \cmidrule(r){6-8} \cmidrule(r){9-11} 
&$(T,n)$ & (50,100) & (100,200) & (200,400)  
& (50,100) & (100,200) & (200,400)    
& (50,100) & (100,200) & (200,400)     \\  
\midrule
     \multirow{4}{*}{$\frac{n}{T}=2$ }    &$RLM$& 14.3  & 18.2  & 55.95  & 11.15 &  15.55 &  47.65 &8.4   & 15.85  & 77.55 \\
         &$RLM_{PE}$ &14.1   & 19.85   & 75.55  & 11.1 & 15.6  & 64.25  & 9.1  &  17.1 & 93.3 \\
         &$LM_{adj}$& 10.85  &16 &54.4  &  8.25 &  13.7 &  46.35 &   6.2& 13.75  &75.85 \\
         &$CD_{P}$ & 5.1   & 6.15  & 7.5     &5.4  &6.1   &  7.25 &  5.1 &6.25   &7.6\\
\hline\hline

    \end{tabular}
    \caption{Empirical power of tests in DGP1 for sparse case}
    \label{tab:power_dgp1_sparse}
\end{table}}

\setlength{\tabcolsep}{1mm}{
\begin{table}[]
\renewcommand\arraystretch{0.8}
\scriptsize
    \centering
    \begin{tabular}{ccccccccccc}
    \hline
   \hline
   \multicolumn{11}{c}{$k=2$}\\
   \hline 
    & & \multicolumn{3}{c}{Chi-squared} & \multicolumn{3}{c}{Normal} & \multicolumn{3}{c}{Student-t}   \\
\cmidrule(r){3-5} \cmidrule(r){6-8} \cmidrule(r){9-11} 
&$(T,n)$ & (50,25) & (100,50) & (200,100)  
& (50,25) & (100,50) & (200,100)   
& (50,25) & (100,50) & (200,100)  \\  
\midrule
     \multirow{4}{*}{$\frac{n}{T}=\frac{1}{2}$ }    &$RLM$& 99.3 & 100 &100 &  98.95 &99.9   &100   &99.15   &97.85   &100   \\
         &$RLM_{PE}$ & 99.9 & 100& 100&99.6  &  100  &100   &99.75   &99.4  & 100   \\
         &$LM_{adj}$&99.55 &100 &100 &99.1   &99.9   &100   & 99.2  & 97.9   &100   \\
         &$CD_{P}$& 54.1 &71.6 & 96.9 &55.7  &66.5    & 96.8  &  54.45 & 46.55  & 95.15   \\
\midrule
  & & \multicolumn{3}{c}{Chi-squared} & \multicolumn{3}{c}{Normal} & \multicolumn{3}{c}{Student-t}   \\
\cmidrule(r){3-5} \cmidrule(r){6-8} \cmidrule(r){9-11}  
&$(T,n)$ & (50,50) & (100,100) & (200,200)  
& (50,50) & (100,100) & (200,200)    
& (50,50) & (100,100) & (200,200)    \\  
\midrule
     \multirow{4}{*}{$\frac{n}{T}=1$ }    &$RLM$& 99.9 & 100 & 100 & 91.45 & 100   &  100 & 75.6  & 100  & 100   \\
         &$RLM_{PE}$ & 100 & 100 & 100 & 96.7  &  100  & 100  & 84.95   & 100  & 100   \\
         &$LM_{adj}$&99.9 &100 &100 & 91.6  &100   &100    &76.05    & 100  &100   \\
         &$CD_{P}$ & 59.8  &68.3&97.85&43.1  &76.85   &97.6   &33.9   &82.45   &98.2   \\
         
\midrule
   & & \multicolumn{3}{c}{Chi-squared} & \multicolumn{3}{c}{Normal} & \multicolumn{3}{c}{Student-t}  \\
\cmidrule(r){3-5} \cmidrule(r){6-8} \cmidrule(r){9-11} 
&$(T,n)$ & (50,100) & (100,200) & (200,400)  
& (50,100) & (100,200) & (200,400)    
& (50,100) & (100,200) & (200,400)     \\  
\midrule
     \multirow{4}{*}{$\frac{n}{T}=2$ }    &$RLM$& 99.8&100 &100 &97.95  &100   &100   & 97.45  & 100  &100   \\
         &$RLM_{PE}$& 100 & 100 & 100 & 99.65  &  100  & 100  & 99.45   & 100  & 100    \\
         &$LM_{adj}$& 99.75 & 100 & 100 & 97.85  &  100  & 100  & 97.45  & 100  & 100  \\
         &$CD_{P}$& 60.85  &79.15&96.9& 51  & 80.2  &94.75   &49.9   &78   &98.6    \\
         
          \hline
   \hline
  \multicolumn{11}{c}{$k=4$}\\
   \hline 
     & & \multicolumn{3}{c}{Chi-squared} & \multicolumn{3}{c}{Normal} & \multicolumn{3}{c}{Student-t}   \\
\cmidrule(r){3-5} \cmidrule(r){6-8} \cmidrule(r){9-11} 
&$(T,n)$ & (50,25) & (100,50) & (200,100)  
& (50,25) & (100,50) & (200,100)   
& (50,25) & (100,50) & (200,100)  \\  
\midrule
     \multirow{4}{*}{$\frac{n}{T}=\frac{1}{2}$ }    &$RLM$& 99.4&99.95 &100 &86.9  & 99.95  &100   & 96.05  & 100  &100   \\
         &$RLM_{PE}$& 99.7 & 100 & 100 & 92.65  &  100  & 100  & 98.35   & 100  & 100    \\
         &$LM_{adj}$& 99.4 & 99.95 & 100 & 87  &  99.95  & 100  & 96.1 & 100  & 100  \\
         &$CD_{P}$& 56.55  &64.8&95.15& 38.7  & 62.5  &95.75   &46.15   &87.2  &93.05   \\
\midrule
  & & \multicolumn{3}{c}{Chi-squared} & \multicolumn{3}{c}{Normal} & \multicolumn{3}{c}{Student-t}  \\
\cmidrule(r){3-5} \cmidrule(r){6-8} \cmidrule(r){9-11}  
&$(T,n)$ & (50,50) & (100,100) & (200,200)  
& (50,50) & (100,100) & (200,200)    
& (50,50) & (100,100) & (200,200)    \\  
\midrule
     \multirow{4}{*}{$\frac{n}{T}=1$ }    &$RLM$& 87.8&100 &100 &78.65  &100   &100   & 94.85 & 99.75 &100   \\
         &$RLM_{PE}$& 94.45 & 100 & 100 & 87.75  &  100  & 100  & 98.4   & 100  & 100    \\
         &$LM_{adj}$& 86.45 & 100 & 100 & 77.65  &  100  & 100  & 94.15 & 99.75 & 100  \\
         &$CD_{P}$& 39.7  &75.2&96.85& 34.5& 71 &93.6  &45.25  &58   &96.35   \\
         
\midrule
  & & \multicolumn{3}{c}{Chi-squared} & \multicolumn{3}{c}{Normal} & \multicolumn{3}{c}{Student-t}  \\
\cmidrule(r){3-5} \cmidrule(r){6-8} \cmidrule(r){9-11} 
&$(T,n)$ & (50,100) & (100,200) & (200,400)  
& (50,100) & (100,200) & (200,400)    
& (50,100) & (100,200) & (200,400)     \\  
\midrule
     \multirow{4}{*}{$\frac{n}{T}=2$ }  &$RLM$& 97.4&100 &100 &95.25  &100   &100   & 91.15  & 100  &100   \\
         &$RLM_{PE}$& 99.45 & 100 & 100 & 98.6  &  100  & 100  & 96.65   & 100  & 100    \\
         &$LM_{adj}$& 96.55 & 100 & 100 & 93.45  &  100  & 100  & 88.4 & 100  & 100  \\
         &$CD_{P}$& 46.15 &74.4 &98.4& 44.15 & 74.9&93.3  &40.75   &80.35 &97.75    \\
\hline\hline

    \end{tabular}
    \caption{Empirical power of tests in DGP1 for less sparse case}
    \label{tab:power_dgp1_less}
\end{table}}

\subsubsection{DGP2: Heterogeneous panel data model with weakly exogenous regressors}
To investigate the performances of the $RLM$ and $RLM_{PE}$ tests in panel data models with weakly exogenous regressors, we consider the following DGP:
\begin{equation*}
    y_{it}=\alpha_i+x_{1,it}\beta_{1i}+x_{2,it}\beta_{2i}+v_{it}, \;\;\;\; i=1,2,\dots, n; \;\;\; t=1,2,\dots, T,
\end{equation*}
where $\alpha_i\sim IIDN(1,1)$, $\beta_{li}\sim IIDN(1,0.04)$, and
\begin{equation*}
    x_{1,it}=0.6x_{1,it-1}+u_{1,it}, \;\;\;\; i=1,2,\dots, n; \;\;\; t=-50,\dots,0,\dots, T,
\end{equation*}
\begin{equation*}
    x_{2,it}=y_{it-1}+u_{2,it}, \;\;\;\; i=1,2,\dots, n; \;\;\; t=-50,\dots,0,\dots, T,
\end{equation*}
with $y_{i,-51}=x_{1,i,-51}=0$ where  $u_{1it}, u_{2it}\sim IIDN(0, \tau_{li}^2/(1-0.6^2))$, $\tau_{li}^2\sim IID \chi^2(6)/6$. This
set up allows for feedback from $y_{it-1}$ to the regressors, thus rendering weakly exogenous. The errors, $\{v_{it}\}$, are generated in the same way as DGP1.

\setlength{\tabcolsep}{1mm}{
\begin{table}[]
\scriptsize
    \centering
    \begin{tabular}{ccccccccccc}
    \hline
   \hline
   \multicolumn{11}{c}{Weakly exogenous}\\
   \hline 
    & & \multicolumn{3}{c}{Chi-squared} & \multicolumn{3}{c}{Normal} & \multicolumn{3}{c}{Student-t}  \\
\cmidrule(r){3-5} \cmidrule(r){6-8} \cmidrule(r){9-11} 
&$(T,n)$ & (50,25) & (100,50) & (200,100)  
& (50,25) & (100,50) & (200,100)   
& (50,25) & (100,50) & (200,100)  \\  
\midrule
     \multirow{4}{*}{$\frac{n}{T}=\frac{1}{2}$ }    &$RLM$& 5.85  & 4.45   & 4.75  & 5.1   & 4.7  & 5.55   & 4.7  & 5.2  &  5.8    \\
         &$RLM_{PE}$ & 5.15   & 4.15  & 4.9  & 4.75  & 4.9  & 6.4 & 5.3  & 5.05  &  5.6 \\
         &$LM_{adj}$& 9.25  & 8.4  & 8.6  & 8.8  & 8.8  & 9.2 & 8.95 & 8.85   &  9.2 \\
         &$CD_{P}$& 4.65  & 4.25   & 4 & 5   & 4.5     & 5.05  & 5.6  & 4.4  &5.35    \\
\midrule
   & & \multicolumn{3}{c}{Chi-squared} & \multicolumn{3}{c}{Normal} & \multicolumn{3}{c}{Student-t}  \\
\cmidrule(r){3-5} \cmidrule(r){6-8} \cmidrule(r){9-11}  
&$(T,n)$ & (50,50) & (100,100) & (200,200)  
& (50,50) & (100,100) & (200,200)    
& (50,50) & (100,100) & (200,200)    \\  
\midrule
     \multirow{4}{*}{$\frac{n}{T}=1$ }    &$RLM$& 4.55  & 5.4 & 5.8  & 5.65  & 4.95  & 5.05   & 5.55  & 5  &5.9     \\
         &$RLM_{PE}$ & 5.2 & 5.05  & 5.15  & 5.6   & 4.95  & 4.9   & 5.85  & 4.9   & 5.05  \\
         &$LM_{adj}$& 13.15   & 13.6  & 14  & 13.65 & 13   & 12.6     & 13.75  &13.4  &  12.6  \\
         &$CD_{P}$ & 5.6  & 4.35  & 5.2   & 6.25  & 5.2   & 5.6 & 6.1 & 5.2  & 6.2    \\
         
\midrule
  & & \multicolumn{3}{c}{Chi-squared} & \multicolumn{3}{c}{Normal} & \multicolumn{3}{c}{Student-t} \\
\cmidrule(r){3-5} \cmidrule(r){6-8} \cmidrule(r){9-11} 
&$(T,n)$ & (50,100) & (100,200) & (200,400)  
& (50,100) & (100,200) & (200,400)    
& (50,100) & (100,200) & (200,400)     \\  
\midrule
     \multirow{4}{*}{$\frac{n}{T}=2$ }    &$RLM$& 5.85  & 6.9   & 5.3  & 5.05   & 5.45  & 4.8   & 4.4   & 5.2   &  5.2   \\
         &$RLM_{PE}$ & 5.05   & 5.75   & 5.1  & 5.5  & 5.75  & 4.7   & 5.7  & 5.2   &  5.6  \\
         &$LM_{adj}$& 28.2  & 27.25  & 27.3 & 27.35  & 28.25   & 26.3 & 30.35 & 28.5 &  26.7   \\
         &$CD_{P}$ & 4.5   & 5.1 & 5.1     & 5.7  & 5.05   & 5.65   & 4.85   & 5.1  &  5.9  \\
\hline\hline

    \end{tabular}
    \caption{Empirical size of tests in DGP2}
    \label{tab:size_dgp2}
\end{table}}

\subsubsection{DGP3: Fixed effects panel model}
The third DGP considered is a fixed effects panel data model with homogeneous coefficients, which is specified as
\begin{equation*}
    y_{it}=\alpha+\sum_{l=2}^kx_{lit}\beta_{l}+\mu_i+v_{it}, \;\;\;\; i=1,2,\dots, n; \;\;\; t=1,2,\dots, T,
\end{equation*}
where $\alpha$ and $\beta_{l}$ are set arbitrarily to 1 and $l$, respectively, $\mu_i\sim IIDN(1,1)$. The regressors and errors are generated in the same way as DGP1.

\setlength{\tabcolsep}{1mm}{
\begin{table}[]
\renewcommand\arraystretch{0.8}
\scriptsize
    \centering
    \begin{tabular}{ccccccccccc}
    \hline
   \hline
   \multicolumn{11}{c}{$k=2$}\\
   \hline 
    & & \multicolumn{3}{c}{Chi-squared} & \multicolumn{3}{c}{Normal} & \multicolumn{3}{c}{Student-t}  \\
\cmidrule(r){3-5} \cmidrule(r){6-8} \cmidrule(r){9-11} 
&$(T,n)$ & (50,25) & (100,50) & (200,100)  
& (50,25) & (100,50) & (200,100)   
& (50,25) & (100,50) & (200,100)  \\  
\midrule
     \multirow{4}{*}{$\frac{n}{T}=\frac{1}{2}$ }     &$RLM$& 5.7 & 4.7  & 5.05  & 5.1 & 5.45 &  5.55 & 5 &5.1  &  4.95 \\
         &$RLM_{PE}$ & 5.55   &  4.85  &  4.95 & 4.9 & 4.8  & 5.6  & 5.2 & 4.7 & 4.6 \\
         &$LM_{adj}$& 6.4 &   5.15&  5.1     & 5.6 & 5.75 &5.65   &5.7  & 5.4 &5.1  \\
         &$CD_{P}$&4.5 & 4.7  & 5.05  & 4.8 & 4.45 & 5.3  & 5.1 & 4.45 &   4.85 \\
\midrule
   & & \multicolumn{3}{c}{Chi-squared} & \multicolumn{3}{c}{Normal} & \multicolumn{3}{c}{Student-t}   \\
\cmidrule(r){3-5} \cmidrule(r){6-8} \cmidrule(r){9-11}  
&$(T,n)$ & (50,50) & (100,100) & (200,200)  
& (50,50) & (100,100) & (200,200)    
& (50,50) & (100,100) & (200,200)    \\  
\midrule
     \multirow{4}{*}{$\frac{n}{T}=1$ }    &$RLM$& 4.85  &  5.25 & 4.6 & 5 & 5 &  4.7 &  5.05& 4.35 &  5.6 \\
         &$RLM_{PE}$ & 4.9  & 4.9      & 4.75 & 5.25 &  4.7 & 5.05& 5.5 & 4 & 5.65   \\
         &$LM_{adj}$& 5.45  & 5.4  &  4.7 & 5.2 & 5.05 &  4.95 &  5.4& 4.55 & 5.65  \\
         &$CD_{P}$ & 5.2 &  4.75 &  6.2 & 5.45 & 5.5 & 5.7  &  5.55& 5.6 &   5.5 \\
         
\midrule
   & & \multicolumn{3}{c}{Chi-squared} & \multicolumn{3}{c}{Normal} & \multicolumn{3}{c}{Student-t}  \\
\cmidrule(r){3-5} \cmidrule(r){6-8} \cmidrule(r){9-11} 
&$(T,n)$ & (50,100) & (100,200) & (200,400)  
& (50,100) & (100,200) & (200,400)    
& (50,100) & (100,200) & (200,400)     \\  
\midrule
     \multirow{4}{*}{$\frac{n}{T}=2$ }    &$RLM$& 6.45  & 6.35  & 4.35  & 5.05 & 5.3 &  4.7 & 4.65 & 5 &  4.3   \\
         &$RLM_{PE}$ & 6.55  & 5.9  & 4.75  & 5.65 & 5.4 & 5.2  & 5 & 5.1 &  4.55 \\
         &$LM_{adj}$& 6.45&  6.35 & 4.35 & 5.15 &  5.35& 4.7  & 4.7 & 5&  4.3 \\
         &$CD_{P}$ & 4.85&  5.95 & 4.9  & 4.95 & 5.4 &  5.55 &  4.25& 5.45 &  5 \\
         
          \hline
   \hline
  \multicolumn{11}{c}{$k=4$}\\
   \hline 
    & & \multicolumn{3}{c}{Chi-squared} & \multicolumn{3}{c}{Normal} & \multicolumn{3}{c}{Student-t}  \\
\cmidrule(r){3-5} \cmidrule(r){6-8} \cmidrule(r){9-11} 
&$(T,n)$ & (50,25) & (100,50) & (200,100)  
& (50,25) & (100,50) & (200,100)   
& (50,25) & (100,50) & (200,100)  \\  
\midrule
     \multirow{4}{*}{$\frac{n}{T}=\frac{1}{2}$ }    &$RLM$& 5.1 &   6.05& 4.7  &  4.7 & 4.65 & 5.35  & 5.05 & 4.45 &   5.45  \\
         &$RLM_{PE}$ & 4.95  & 5.05  & 5.05  & 4.45 & 4.25 & 5.05  & 4.2 & 4.85 &  5.25 \\
         &$LM_{adj}$& 5.4  &  6.4 &  4.75 & 5 & 4.7 &  5.35 & 5.35 & 4.5 & 5.45 \\
         &$CD_{P}$& 4.95 & 6.1  &  4.6 &  4.65 & 5.2 & 5.25  & 4.35 & 5.05&  5.05 \\
\midrule
  & & \multicolumn{3}{c}{Chi-squared} & \multicolumn{3}{c}{Normal} & \multicolumn{3}{c}{Student-t} \\
\cmidrule(r){3-5} \cmidrule(r){6-8} \cmidrule(r){9-11}  
&$(T,n)$ & (50,50) & (100,100) & (200,200)  
& (50,50) & (100,100) & (200,200)    
& (50,50) & (100,100) & (200,200)    \\  
\midrule
     \multirow{4}{*}{$\frac{n}{T}=1$ }    &$RLM$& 5.75  & 5.15  & 4.75  &  4.8 & 5.25 &  4.75 & 5 & 5.8 &  5.3  \\
         &$RLM_{PE}$ & 5.7  & 5.05  & 4.85  & 4.7 & 5.7 &  4.75 &  5.2& 5.45 &5.2    \\
         &$LM_{adj}$& 5.4  & 4.8  & 4.55  & 4.2 & 5.05 & 4.6  &  4.5& 5.55 &  5.25   \\
         &$CD_{P}$ & 5.05&  4.9 & 5.65  & 4.95  & 4.85 & 5  & 5.3& 5.35 &  4.9  \\
         
\midrule
  & & \multicolumn{3}{c}{Chi-squared} & \multicolumn{3}{c}{Normal} & \multicolumn{3}{c}{Student-t}  \\
\cmidrule(r){3-5} \cmidrule(r){6-8} \cmidrule(r){9-11} 
&$(T,n)$ & (50,100) & (100,200) & (200,400)  
& (50,100) & (100,200) & (200,400)    
& (50,100) & (100,200) & (200,400)     \\  
\midrule
     \multirow{4}{*}{$\frac{n}{T}=2$ }    &$RLM$& 5.35  & 6.15  & 5.1  & 6.05 & 5.4 &  4.75 &  4.9&  5.3&  4.9  \\
         &$RLM_{PE}$ & 5.65 & 5.8   & 4.95  & 5.55 & 5.1 &  5.05 & 5.55 & 5.1 &5.1  \\
         &$LM_{adj}$&4.15 &  5.25 &  4.85 & 4.55 & 4.3 & 4.4  & 3.95 & 4.9 & 4.4  \\
         &$CD_{P}$ & 5  & 4.9  & 3.75   & 5.15 & 5.15 &  5.15 & 5.3 & 5.4 &   4.7\\
\hline\hline

    \end{tabular}
    \caption{Empirical size of tests in DGP3}
    \label{tab:size_dgp3}
\end{table}}

\subsubsection{DGP4: Dynamic panel data model}
To examine the properties of the $RLM$ and $RLM_{PE}$ tests in a dynamic panel data model, we follow the design  of \cite{pesaran2008bias}:
\begin{equation*}
    y_{it}=\xi_i(1-\beta_i)+\beta_iy_{it-1}+v_{it} \;\;\;\; i=1,2,\dots, n; \;\;\; t=-50,\dots,0,\dots, T,
\end{equation*}
with $y_{i,-51}=0$, where $\beta_{i}\sim IIDN(1,0.04)$, and the fixed effects, $\xi_i$, are drawn as $v_{i0}+\eta_i$, with $\eta_i \sim IIDN(1,2)$. The errors, $\{v_{it}\}$,  are  generated in the same way as DGP1.

\setlength{\tabcolsep}{1mm}{
\begin{table}[]
\scriptsize
    \centering
    \begin{tabular}{ccccccccccc}
    \hline
   \hline
   \multicolumn{11}{c}{Dynamic}\\
   \hline 
    & & \multicolumn{3}{c}{Chi-squared} & \multicolumn{3}{c}{Normal} & \multicolumn{3}{c}{Student-t} \\
\cmidrule(r){3-5} \cmidrule(r){6-8} \cmidrule(r){9-11} 
&$(T,n)$ & (50,25) & (100,50) & (200,100)  
& (50,25) & (100,50) & (200,100)   
& (50,25) & (100,50) & (200,100)  \\  
\midrule
     \multirow{4}{*}{$\frac{n}{T}=\frac{1}{2}$ }    &$RLM$& 6.3  & 5.25  &  5.1 &  5.05 & 5.5  & 5.25 & 5.4  &  4.75 & 4.85  \\
         &$RLM_{PE}$ & 5.8    & 5.15  & 4.45  &   4.7& 5.3 & 4.9 &   5.4&  4.7 & 4.9  \\
         &$LM_{adj}$& 6.75  & 5.45  & 5.1  &  5.45 & 5.6 & 5.3 &5.85   & 5.15  & 4.95\\
         &$CD_{P}$& 4.9  &  5.3 &  5.4 &  4.35 & 5 & 5.45 & 4.35  &  5.8 & 5.15 \\
\midrule
  & & \multicolumn{3}{c}{Chi-squared} & \multicolumn{3}{c}{Normal} & \multicolumn{3}{c}{Student-t}   \\
\cmidrule(r){3-5} \cmidrule(r){6-8} \cmidrule(r){9-11}  
&$(T,n)$ & (50,50) & (100,100) & (200,200)  
& (50,50) & (100,100) & (200,200)    
& (50,50) & (100,100) & (200,200)    \\  
\midrule
     \multirow{4}{*}{$\frac{n}{T}=1$ }    &$RLM$& 6.65   &5.25   & 5.95  & 5.8  & 5.35 & 4.6 & 6.2  & 5.3  & 5  \\
         &$RLM_{PE}$ & 5.7  & 5.35  & 5.4  & 6.1  &  5.45&4.45  &  6.05 & 5.5  & 4.6 \\
         &$LM_{adj}$& 6.7  & 5.3  &  6 &  5.9 & 5.35 & 4.6 & 6.35  &  5.4 & 5\\
         &$CD_{P}$ & 4.75   &  4.85 & 4.85  & 4.85  & 4.6 & 4.25 &   4.8 & 4.8  &5.5 \\
         
\midrule
  & & \multicolumn{3}{c}{Chi-squared} & \multicolumn{3}{c}{Normal} & \multicolumn{3}{c}{Student-t} \\
\cmidrule(r){3-5} \cmidrule(r){6-8} \cmidrule(r){9-11} 
&$(T,n)$ & (50,100) & (100,200) & (200,400)  
& (50,100) & (100,200) & (200,400)    
& (50,100) & (100,200) & (200,400)     \\  
\midrule
     \multirow{4}{*}{$\frac{n}{T}=2$ }    &$RLM$& 5.7    & 5.9  &5.65   &6.2  & 5.1 & 5.35 &  5.95 & 5.5 & 5.45\\
         &$RLM_{PE}$ & 6.45  & 6.05  &  5.1 & 5.4  & 5.45 & 5.05 & 5.05 &  5.3 & 5.15 \\
         &$LM_{adj}$& 5.3  &  5.7 &  5.6 &  5.75 & 5 &  5.3&  5.6 &  5.35 & 5.3\\
         &$CD_{P}$ & 4.4  & 6  &  4.55 & 5.7  & 5.75 &  5.05&  5 & 4.95  &  5.6 \\
\hline\hline

    \end{tabular}
    \caption{Empirical size of tests in DGP4}
    \label{tab:size_dgp4}
\end{table}}

\subsection{Simulation results}
Table \ref{tab:DGP1} reports the empirical size of these tests for the DGP1. The proposed $RLM$ and $RLM_{PE}$ tests successfully control the size under almost all settings, irrespective of number of regressors included in the panel data model\footnote{However, for small sample size with more regressors, i.e. $T=50$ and $k=4$, the $RLM$ and $RLM_{PE}$ tests would be slightly oversized. For example, the empirical sizes of $RLM$ and $RLM_{PE}$ are 7.65 and 7.4 under normal errors, respectively. }.  For a fixed ratio $n/T$, the empirical size of the $RLM$ and $RLM_{PE}$ tests converge to the nominal size of $0.05$ as $T\rightarrow\infty$, that authenticates the asymptotic normality of the tests under the SIM-L scheme. Besides, the performance of $RLM$ are almost identical to $LM_{adj}$.  The $CD_P$ test has correct size in all cases. 

Table \ref{tab:power_dgp1_dense2} demonstrates the empirical power of these tests under the alternative with dense factors. The $RLM$ test has comparable power to $LM_{adj}$ regardless of $(n,T)$ combinations and error distributions. In contrast, the $CD_P$ test suffers from little power by construction, where mean of factor loading is close to zero  as mentioned by \cite{pesaran2008bias}. The power enhancement version of $RLM$, the $RLM_{PE}$ test, outperforms others across the board, especially when $n/T=2$. For example, the power of the  $RLM_{PE}$ test is 78.25\% for $T=100, n=200, k=4, h=3$ and student-t errors, whereas the power of the $RLM$ and $LM_{adj}$ tests are 60.1\% and 56.75\%, respectively. It improves the power by up to around 30\%. Besides,  the power of the $RLM_{PE}$ test is 69.2\% for $T=100, n=200, k=2, h=3$ and chi-square errors, and the power of the $RLM$ and $LM_{adj}$ tests are 51.2\% and 51.1\%, respectively. These results indicate that the $RLM_{PE}$ test successfully boost the power. 

The empirical power of these tests under the alternative with sparse and less sparse factors are summarized in Tables \ref{tab:power_dgp1_sparse} and \ref{tab:power_dgp1_less}, respectively. The $RLM$ test again has similar performance to $LM_{adj}$. The empirical power of $RLM_{PE}$ show that it performs the best among those tests. The power of the $CD_P$ test floats around 5\% throughout as in \cite{pesaran2015testing}.

For the heterogeneous panel data model with weakly exogenous regressors, Table \ref{tab:size_dgp2} shows that $LM_{adj}$ becomes considerably oversized, especially for the case $n/T=2$, where it has size around 28\%. This shows that the $LM_{adj}$ test is not robust to weakly exogenous regressors, which is also observed in  \cite{bailey2021lagrange}. However, our proposed $RLM$ test and $RLM_{PE}$ tests control the size well, which are not sensitive to the strictly exogenous assumptions on regressors. Therefore, though it is widely reported that the $LM_{adj}$ test has generally satisfying empirical performances regardless of restrictive strictly exogenous regressors and normal errors assumptions, the present DGP2 is indeed a rare situation where the $LM_{adj}$ test is outperformed by others.

Table \ref{tab:size_dgp3} reports the size of the tests for the fixed effects panel data model. It shows that the
proposed $RLM$ and $RLM_{PE}$ tests have the correct size, close to the 5\% nominal significance level, for example, $RLM$ has 5.1\% and 5\% size results, respectively for $T=50,n=25$ under normal errors and for $T=100,n=200$ under chi-squared errors. Similar results for $RLM_{PE}$ can be also observed in this table. Pesaran's $LM_{adj}$ and $CD_P$ tests have correct size in this setting as in \cite{pesaran2004general} and \cite{pesaran2008bias}.

Finally, Table \ref{tab:size_dgp4} gives the empirical size of these tests for dynamic panel data model. It shows that the proposed $RLM$ and $RLM_{PE}$ tests have the correct size, e.g. 5.15\% for $n=100, T=100$ with chi-squared error, which is comparable to the $LM_{adj}$ test. The $CD_P$ always has correct size as in \cite{pesaran2004general}. The results of empirical power for DGP2, DGP3 and DGP4 are similar to those of DGP1, so we omit it here.

Based on these findings, the $RLM_{PE}$ test is strongly recommended for practitioners if it is not clear whether weakly exogenous regressors are present or not, given its universally correct size and better power performances. Instead, when regressors are believed to be strictly exogenous, then the $RLM_{PE}$ test, a easily implemented and computationally cheap procedure, is preferred for large panels ($T\geq 50$). For $T\leq 50$, $RLM_{PE}$ is still applicable though it might be slightly oversized, or the $LM_{adj}$ test is a suggested method at the risk of intensive computation.

\section{Conclusion}
\label{sec:conclusion}
This paper has developed a Lagrange multiplier type test for the null hypothesis of no cross-sectional dependence in large panel models. The procedure can be applied to a wide class of linear panel data models and shows robustness to quite general forms of non-normality in the disturbance distribution. We further proposed a power enhancement version of the $LM$ type test based on the fourth moment of the sample correlations obtained from residuals to boost power under sparse alternatives, which only requires existence of higher moment but still shares such robustness. The simulations illustrate that this test has satisfactory power under the sparse alternatives of weak cross-sectional dependence, and both of the tests successfully control the size in different data generating processes. 

For future work, it is interesting to explore theoretically the power properties of $RLM_{PE}$  and the optimal $m$ that would maximise power. Also, it would be of interest to investigate the performance of the $RLM$ and $RLM_{PE}$ tests in the weakly cross-sectional dependence framework. In addition, testing the null hypothesis with no cross-sectional dependence when errors are serial dependent will also be studied.

\appendix
\newpage
\markboth{Appendix}{Appendix}
\renewcommand{\thesection}{\Alph{section}}
\numberwithin{equation}{section}
\section*{Appendix}
This appendix includes the proofs of the following lemmas:
\begin{lemma} \label{clt:hete}
  Under Assumptions 1, 2 and 3$(i)$, 
$$\frac{tr(\R^2)-\mu_0}{\sigma_0}\cvd N(0,1).$$
\end{lemma}

\begin{lemma} \label{op1:hete}Under Assumptions 1, 2, 3$(i)$ and 4,
  $$tr(\hat{\R}^2)=tr(\R^2)+o_p(1).$$
\end{lemma}

\begin{lemma} \label{clt:hete:4}
  Under Assumptions 1, 2 and 3$(ii)$,
$$\frac{tr(\R^4)-\mu_{PE}}{\sigma_{PE}}\cvd N(0,1).$$
\end{lemma}

\begin{lemma} \label{op1:hete:4} Under Assumptions 1, 2, 3$(ii)$ and 4,
  $$tr(\hat{\R}^4)=tr(\R^4)+o_p(1).$$
\end{lemma}

In the static heterogeneous panel data model, $\hab_i=\left(\X_i\X_i^{\prime}\right)^{-1}\X_i\Y_i$ is the OLS estimator and the residuals are given by $\vv_{it}=v_{it}-\x_{it}^{\prime}\left(\hab_i-\bbeta_i\right)$. Let $\bv_i=(v_{i1},\dots,v_{iT})'$, $\hat{\bv}_i=(\vv_{i1},\dots,\vv_{iT})'$ for $i=1,\dots, n$, consequently, $\hat{\bv}_i=\bv_i-\X_i^{\prime}\left(\hab_i-\bbeta_i\right)$. Define $\mathbf{\VVV}=\begin{pmatrix} \hbv_{1} , \cdots,  \hbv_{n} \end{pmatrix}$, $\bw_i=\X_i^{\prime}\left(\hab_i-\bbeta_i\right)$ and $\mathbf{W}=\begin{pmatrix} \bw_{1} , \cdots,  \bw_{n} \end{pmatrix}$. Using this notation,  $\hbv_i=\bv_i-\bw_i$, $\mathbf{\VVV}=\mathbf{V}-\mathbf{W}$ and the sample covariance matrices can be written as  $\mathbf{S}_T=\frac1T\mathbf{\VVV}^{\prime}\mathbf{\VVV}$, $\hat{\mathbf{S}}_T=\frac1T\mathbf{V}^{\prime}\mathbf{V}$ with elements
$S_{T,i,j}=\frac1T\sum\limits_{t=1}^{T}v_{it}v_{jt}$, $\hat{S}_{T,i,j}=\frac1T\sum\limits_{t=1}^{T}\vv_{it}\vv_{jt},$ respectively.

To accomplish the proof of results above, several lemmas are introduced as following.

\begin{lemma} (Theorem 13 of Chapter 13, \cite{petrov1975sum}) \label{lem:petrov}
  Let $Y_1, \cdots, Y_n$ be independent and identically distributed random variables, such that $E(Y_1)=0$, $E(Y_1)^2=1$ and $E|Y_1|^r<\infty$ for some $r\geq3$. Then
  $$\Big|\mathbb{P}(\frac{1}{\sqrt{n}}\sum_{i=1}^nY_i<y)-\Phi(y)\Big|\leq \frac{C(r)}{(1+|y|)^r}\left[\frac{E|Y_1|^3}{\sqrt{n}}+ \frac{E|Y_1|^r}{n^{\frac{r-2}{2}}}\right]$$ for all $y$, where $\Phi(\cdot)$ is the cumulative distribution function of standard normal random variable and $C(r)$ is a positive constant depending only on $r$.
\end{lemma}

\begin{lemma}\label{lem:max_order}
 Let $\{Y_{ij}\}_{i\geq1,j\geq 1}$ be an array of independent and  identically distributed random variables such that $E(Y_{11})=0$, $E(Y_{11})^2=1$ and $E|Y_{11}|^r<\infty$ for some $r\geq3$. Let $X_{in}=\frac{1}{\sqrt{n}}\sum_{j=1}^nY_{ij}$, then for any  $\epsilon >0$, we have 
 $$\max_{1\le i \le n}|X_{in}|=O_p(n^{\frac{1}{2r}+\epsilon}).$$
\end{lemma}
\begin{proof}
    For some $\alpha>0,$ let $c=\frac{1}{2r}+\epsilon$, we have 
     \begin{align*}
    \mathbb{P}(\max_{1\le i \le n}|X_{in}|\geq \alpha n^c) & = 1-\left[ \mathbb{P}(|X_{1n}|<\alpha n^c) \right]^n\\
    & \le 1- \Big\{\Phi( \alpha n^c) -\frac{C(r)}{(1+ \alpha n^c)^r} \left[\frac{E|Y_{11}|^3}{\sqrt{n}}+ \frac{E|Y_{11}|^r}{n^{\frac{r-2}{2}}}\right] \Big\}^n \\
    &\sim 1- \Big\{1-\frac{1}{\sqrt{2\pi}\alpha}n^{-c}e^{-\frac{\alpha^2n^{2c}}{2}}
    -\frac{C(r)}{(1+ \alpha n^c)^r} \left[\frac{E|Y_{11}|^3}{\sqrt{n}}+ \frac{E|Y_{11}|^r}{n^{\frac{r-2}{2}}}\right] \Big\}^n \\
    &\sim \frac{1}{\sqrt{2\pi}\alpha}n^{1-c}e^{-\frac{\alpha^2n^{2c}}{2}}+ \frac{C(r)E|Y_{11}|^3\sqrt{n}}{(1+ \alpha n^c)^r}+ \frac{C(r)E|Y_{11}|^r}{(1+ \alpha n^c)^rn^{\frac{r}{2}-2}},
  \end{align*}
  where the first inequality follows by Lemma \ref{lem:petrov}, and the first approximation follows by the fact that $\Phi(x) \sim 1- \frac{1}{\sqrt{2\pi}x}e^{-\frac{x^2}{2}}$ for large $x$. Therefore, $\max_{1\le i \le n}|X_{in}|=O_p(n^c)$ holds.
\end{proof}

\begin{remark}
 For the panel data model, if the errors $\{v_{it}\}_{1\leq i \leq n, 1\leq t \leq T}$ satisfy the conditions in lemma \ref{lem:petrov} and $X_{iT}=\frac{1}{\sqrt{T}}\sum_{t=1}^Tv_{it}$, then for any any  $\epsilon >0$, the estimate $\max_{1\le i \le n}|X_{iT}|=O_p(n^{\frac{1}{2r}+\epsilon})$ still holds once we further assume that $K_1\leq\frac{n}{T}\leq K_2$ for some positive constants $K_1$ and $K_2$. It holds naturally since $\frac{n}{T}\rightarrow c>0$ in the SIM-L scheme.
\end{remark}

\begin{lemma} (\cite{li2012jiang}) \label{lem:max_rho} 
  Suppose $E|v_{it}|^{6}<\infty$, then $$\max_{1\leq i<j\leq n} |\rho_{ij}|=O_p\Bigg(\sqrt{\frac{\log n}{n}}\Bigg),$$
  when $\frac nT\rightarrow c >0$.
\end{lemma}

\begin{lemma}\label{lem:estimates}
  Under  
  Assumptions 1, 2, 3 and 4, for any $\epsilon>0$ and some integer $r_1\geq 3$, i.e. $E|v_{it}|^{r_1}<\infty$
  \begin{enumerate}
  \item[(a)\ ]  $\displaystyle \max_{1\le i ,  j \le n}    |\inn{\bv_i}{\bw_j}|=O_p (n^{\frac{1}{r_1}+2\epsilon_1})$.
  \item[(b)\ ]  $\displaystyle \max_{1\le i, j \le n}    |\inn{\bw_i}{\bw_j}|=O_p (n^{\frac{1}{r_1}+2\epsilon_1})$.
    
  \item[(c)\ ]  $\displaystyle \max_{1\le i  \le n} |S_{T,i,i} -\sigma^2|= O_p (n^{\frac{1}{r_1}+\epsilon_2-\frac12})$.
  \item[(d)\ ]  $\displaystyle \max_{1\le i \ne j \le n} |S_{T,i,j} |= O_p (n^{\frac{1}{r_1}+\epsilon_2-\frac12})$.
  \item[(e)\ ]
    $\displaystyle   \max_{1\le i  \le n} |   \hat S_{T,i,i} -    S_{T,i,i} |   =O_p (n^{\frac{1}{r_1}+2\epsilon_1-1}).$
  \item[(f)\ ]  $\displaystyle \max_{1\le i  \le n} |\hat S_{T,i,i} -\sigma^2|= O_p (n^{\frac{1}{r_1}+\epsilon_2-\frac12})$.

  \end{enumerate}
\end{lemma}

\begin{proof}

  (a). Firstly, we consider the case $i=j$. By Assumption 2, we obtain 
  $\boldsymbol{\xi}_i = \left( \X_i \X_i' \right)^{-\frac12} \sum\limits_{t}\x_{it}v_{it}\cvd N_k\left(0,\sigma_i^2\mathbf{I}_k\right).
  $
  Therefore, by Lemma \ref{lem:max_order} for some $r_1 \geq 3$ and for any $\epsilon_1>0$, 
  \[ \max_{1\le i\le n} \|\boldsymbol{\xi}_i\| =O_p(n^{\frac{1}{2r_1}+\epsilon_1}).
  \]
  Consequently,
  $$\begin{aligned}
     &\max_{1\le i \le n} \left|\inn{\bw_i}{\bv_i}\right|\\&=\max_{1\le i  \le n}\left|\bv_i^{\prime}\X_i^{\prime}\left(\hab_i-\bbeta_i\right)\right|
    = \max_{1\le i  \le n}\left|\bv_i^{\prime}\X_i^{\prime}\left(\X_i\X_i^\prime\right)^
    {-1}\X_i\bv_i\right|
    = \max_{1\le i  \le n}\left\|\boldsymbol{\xi}_i\right\|^2=O_p(n^{\frac{1}{r_1}+2\epsilon_1}).
  \end{aligned}$$
The calculations for $i\ne j$ case is similar so we omit it here.

  \bigskip

  (b). We have
  \begin{align*}
    \max_{1\le i,j  \le n}\left| \inn{\bw_i}{\bw_j} \right| &= \max_{1\le i,j \le n}\left| \left(\hab_i^{\prime}-\bbeta_i^{\prime}\right)\X_i\X_j^{\prime} \left(\hab_j-\bbeta_j\right)\right| \\& = \max_{1\le i,j \le n} \left|\bv_i^{\prime}\X_i^{\prime}\left(\X_i\X_i^\prime\right)^
    {-1}\X_i \X_j^{\prime}\left(\X_j\X_j^\prime\right)^
    {-1}\X_j^{\prime}\bv_j \right| \\
    &=\max_{1\le i , j\le n} \left|\boldsymbol{\xi}_i^\prime \left(\X_i\X_i^\prime\right)^
    {-\frac{1}{2}}\X_i \X_j^{\prime}\left(\X_j\X_j^\prime\right)^
    {-\frac{1}{2}}\boldsymbol{\xi}_j \right|\\
    & \le \left(\max_{1\le i \le n} \|\boldsymbol{\xi}_i\| \| \X_i\X_i'\|^{-1/2} \|\X_i\|\right)^2  \\
    & =  \left(O_p(n^{\frac{1}{2r_1}+\epsilon_1})O_p(T^{-1/2})  O_p(T^{1/2})\right)^2   = O_p(n^{\frac{1}{r_1}+2\epsilon_1}).
  \end{align*}

  \bigskip

  (c). By CLT, $Z_i\stackrel{\Delta}{=} \sqrt T (S_{T,i,i}-\sigma_i^2)\cvd
  N(0,\tau_i^2)$ where $\sigma^2=E(v_{it}^2)$ and $\tau_i^2=\text{var}
  (v_{it}^2)$, then 
  $ \max\limits_{1\le i\le n} |Z_i| = O_p(n^{\frac{1}{2r_2}+\epsilon_2})  $  for some $r_2 \geq 3$ and for any $\epsilon_2>0$ by Lemma \ref{lem:max_order}. It can be easily found that $r_1=2r_2$, therefore,
  \[  \max\limits_{1\le i\le n} |S_{T,i,i}-\sigma^2| = O_p(n^{\frac{1}{2r_2}+\epsilon_2-\frac12})=O_p(n^{\frac{1}{r_1}+\epsilon_2-\frac12}).
  \]

  \bigskip
  
  (d). Note that $E(v_{it}v_{jt})=0$ for $i\ne j$, it follows along the same lines as that of (c).
   
   \bigskip
  
  (e). By (a) and (b), 
  we have 
  \begin{align*}
    \nonumber
    \max_{1\le i  \le n} | \hat S_{T,i,i} -  S_{T,i,i} | & =\max_{1\le i \le n} \left|  -2T^{-1} \inn{\bw_i}{\bv_i} +T^{-1}\|\bw_i\|^2 \right|
      = \frac1T  O_p (n^{\frac{1}{r_1}+2\epsilon_1})  = O_p
     (  n^{\frac{1}{r_1}+2\epsilon_1-1}  ) . 
  \end{align*}

  \bigskip
  (f).  The conclusion holds from (c) and (e). 
  
\end{proof}

\noindent\textbf{Proof of Lemma 1}
\begin{proof}
By the Theorem 3.1 of \cite{Yin2021spectral}, there exist constants $\mu_{center}$, $\mu_{limit}$ and $\sigma_0>0$ such that: $$tr(\R^2)-\mu_{center}\cvd N(\mu_{limit},\sigma_0'^2).$$ Applying the results in Example 3.2 of \cite{Yin2021spectral} with $g_l=x^2$, we obtain $\mu_{center}=n(1+\frac{n}{T-1})$. For the case $g_l=x^2$ and $\R=\mathbf{I}_n$, the results in Example 3.3 of \cite{Yin2021spectral} shows that $\mu_{limit}=-c$ and $\sigma_{0}'=2c$. Finally, substituting $c$ with $c_T$ by Slutsky's theorem completes the proof.
\end{proof}

\noindent\textbf{Proof of Lemma 2}
\begin{proof}
  By direct calculation we have
$$\begin{aligned}
  \left| tr\left(\hat{\mathbf{R}}^2-\mathbf{R}^2\right)\right|&=\left|\sum\limits_{1\leq i\neq j \leq n}  \left(  \rhohat^2 -\rho_{ij}^2\right)\right|\\ &=\left|\sum\limits_{1\leq i\neq j \leq n}\rho_{ij}^2 \left(\frac{\rhohat^2}{\rho_{ij}^2}-1 \right)  \right| \\
  &\le \sqrt{\sum\limits_{1\leq i\neq j \leq n}\rho_{ij}^4} \cdot\sqrt{\sum\limits_{1\leq i\neq j\leq n}\left(\frac{\rhohat^2}{\rho_{ij}^2}-1 \right)^2} \\
  &= \sqrt{\frac1{T^{\alpha_1}}\sum\limits_{1\leq i\neq j \leq n}\rho_{ij}^4}\cdot \sqrt{T^{\alpha_1}\sum\limits_{1\leq i\neq j\leq n}\left( \frac{\Stij^2\hat{S}_{T,i,i}\hat{S}_{T,j,j}-\hStij^2S_{T,i,i}S_{T,j,j}}{\Stij^2\hat{S}_{T,i,i}\hat{S}_{T,j,j}} \right)^2},\\
\end{aligned}$$
where constant $0<\alpha_1<1$. Therefore, we aim to show that: 
\begin{enumerate}
    \item[\textbf{(i)}] $A_1\stackrel{\Delta}{=}\left|\frac1{T^{\alpha_1}}\sum\limits_{1\leq i\neq j \leq n}\rho_{ij}^4\right|=o_p(1)$,
    \item[\textbf{(ii)}]  $A_2\stackrel{\Delta}{=}\left|T^{\alpha_1}\sum\limits_{1\leq i\neq j \leq n}\left( \frac{\Stij^2\hat{S}_{T,i,i}\hat{S}_{T,j,j}-\hStij^2S_{T,i,i}S_{T,j,j}}{\Stij^2\hat{S}_{T,i,i}\hat{S}_{T,j,j}} \right)^2\right|=o_p(1)$.
\end{enumerate}
\textbf{(i)} By lemma \ref{lem:max_rho}, we have $$\left|\frac1{T^{\alpha_1}}\sum\limits_{1\leq i\neq j \leq n}\rho_{ij}^4\right| \leq \frac1{T^{\alpha_1}}O_p\left((\frac{\log n}{n})^2\cdot n^2\right)=O_p\Bigg(\frac{(\logn)^2}{T^{\alpha_1}}\Bigg)=o_p(1).$$ Therefore, $A_1=o_p(1)$ holds.

\noindent \textbf{(ii)} By direct calculation we have 
$$\begin{aligned}
  A_2  &\le T^{\alpha_1-8}\sum\limits_{1\leq i\neq j \leq n}\frac{1}{\Stij^4\hat{S}_{T,i,i}^2\hat{S}_{T,j,j}^2}\Big|(\bv_i'\bv_j)^2(\bv_i'-\bw_i')(\bv_i-\bw_i)(\bv_j'-\bw_j')(\bv_j-\bw_j)-\\
  & \left[(\bv_i'-\bw_i')(\bv_j-\bw_j)\right]^2 \bv_i'\bv_i\bv_j'\bv_j\Big|^2\\
  &=T^{\alpha_1-8}\sum\limits_{1\leq i\neq j \leq n}\frac{1}{\Stij^4\hat{S}_{T,i,i}^2\hat{S}_{T,j,j}^2}\Big| -2(\bv_i'\bv_j)^2\bw_i'\bv_i\bv_j'\bv_j+(\bv_i'\bv_j)^2\bw_i'\bw_i\bv_j'\bv_j-2(\bv_i'\bv_j)^2\bv_i'\bv_i\bw_j'\bv_j\\&+4(\bv_i'\bv_j)^2\bw_i'\bv_i\bw_j'\bv_j-2(\bv_i'\bv_j)^2\bw_i'\bw_i\bw_j'\bv_j+(\bv_i'\bv_j)^2\bv_i'\bv_i\bw_j'\bw_j-2(\bv_i'\bv_j)^2\bw_i'\bv_i\bw_j'\bw_j\\&+(\bv_i'\bv_j)^2\bw_i'\bw_i\bw_j'\bw_j-(\bw_i'\bv_j)^2\bv_i'\bv_i\bv_j'\bv_j-(\bv_i'\bw¬_j)^2\bv_i'\bv_i\bv_j'\bv_j-(\bw_i'\bw_j)^2\bv_i'\bv_i\bv_j'\bv_j\\&+2\bv_i'\bv_j\bw_i'\bv_j\bv_i'\bv_i\bv_j'\bv_j+2\bv_i'\bv_j\bv_i'\bw_j\bv_i'\bv_i\bv_j'\bv_j-2\bv_i'\bv_j\bw_i'\bw_j\bv_i'\bv_i\bv_j'\bv_j-2\bw_i'\bv_j\bv_i'\bw_j\bv_i'\bv_i\bv_j'\bv_j\\&+2\bw_i'\bv_j\bw_i'\bw_j\bv_i'\bv_i\bv_j'\bv_j+2\bv_i'\bw_j\bw_i'\bw_j\bv_i'\bv_i\bv_j'\bv_j\Big|^2.\\
\end{aligned}$$
Define $RHS\stackrel{\Delta}{=} T^{{\alpha_1-8}}\sum\limits_{1\leq i\neq j \leq n}\frac{1}{\Stij^4\hat{S}_{T,i,i}^2\hat{S}_{T,j,j}^2}|\sum\limits_{m=1}^{17}A_{2,m}|^2,$ then
\begin{equation*}
    A_2\leq T^{\alpha_1-8}\sum\limits_{1\leq i\neq j \leq n}\frac{1}{\Stij^4\hat{S}_{T,i,i}^2\hat{S}_{T,j,j}^2}\sum\limits_{ m_1,m_2 }|\tau_{m_1,m_2}A_{2,m_1}A_{2,m_2}|, 
\end{equation*}
where $\tau_{m_1,m_2}$ is a constant only depending on $m_1$ and $m_2$.

Therefore, if $T^{\alpha_1-8}\sum\limits_{1\leq i\neq j \leq n}\frac{|A_{2,m_1}A_{2,m_2}|}{\Stij^4\hat{S}_{T,i,i}^2\hat{S}_{T,j,j}^2}=o_p(1)$ for any $ 1\leq m_1, m_2 \leq 17$ holds, then we can conclude that      $A_2=o_p(1)$.  Further, we only need to consider the case when $m_1=m_2$ by equality $2\cdot|A_{2,m_1}A_{2,m_2}|\leq |A_{2,m_1}|^2+|A_{2,m_2}|^2$, i.e. if we can show $T^{\alpha_1-8}\sum\limits_{1\leq i\neq j \leq n}\frac{|A_{2,m}|^2}{\Stij^4\hat{S}_{T,i,i}^2\hat{S}_{T,j,j}^2}=o_p(1)$ for any $1\leq m \leq 17$, then  $A_2=o_p(1)$ immediately holds. 
By Lemma \ref{lem:estimates}, we show that 
$$\begin{aligned}
T^{\alpha_1-8}\sum\limits_{1\leq i\neq j \leq n}\frac{|A_{2,1}|^2}{\Stij^4\hat{S}_{T,i,i}^2\hat{S}_{T,j,j}^2}&=T^{\alpha_1-8}\sum\limits_{1\leq i\neq j \leq n}\frac{|-2(\bv_i'\bv_j)^2\bw_i'\bv_i\bv_j'\bv_j|^2}{\Stij^4\hat{S}_{T,i,i}^2\hat{S}_{T,j,j}^2}\\&=T^{\alpha_1-8}\sum\limits_{1\leq i\neq j \leq n}|O_p(n^{\frac{3}{r_1}+2\epsilon_1+2\epsilon_2+2})|^2\\&=O_p(n^{\frac{6}{r_1}+4\epsilon_1+4\epsilon_2+\alpha_1-2})=o_p(1).
\end{aligned}$$
By similar calculations, we conclude that $$T^{\alpha_1-8}\sum\limits_{1\leq i\neq j \leq n}\frac{|A_{2,m}|^2}{\Stij^4\hat{S}_{T,i,i}^2\hat{S}_{T,j,j}^2}=O_p(n^{\frac{6}{r_1}+4\epsilon_1+4\epsilon_2+\alpha_1-2}) \quad m=1,2,3,6,$$
$$T^{\alpha_1-8}\sum\limits_{1\leq i\neq j \leq n}\frac{|A_{2,m}|^2}{\Stij^4\hat{S}_{T,i,i}^2\hat{S}_{T,j,j}^2}=O_p(n^{\frac{8}{r_1}+8\epsilon_1+4\epsilon_2+\alpha_1-4}) \quad m=4,5,7,8,$$$$T^{\alpha_1-8}\sum\limits_{1\leq i\neq j \leq n}\frac{|A_{2,m}|^2}{\Stij^4\hat{S}_{T,i,i}^2\hat{S}_{T,j,j}^2}=O_p(n^{\frac{4}{r_1}+8\epsilon_1+\alpha_1-2}) \quad m=9,10,11,15,16,17,$$
$$T^{\alpha_1-8}\sum\limits_{1\leq i\neq j \leq n}\frac{|A_{2,m}|^2}{\Stij^4\hat{S}_{T,i,i}^2\hat{S}_{T,j,j}^2}=O_p(n^{\frac{4}{r_1}+4\epsilon_1+2\epsilon_2+\alpha_1-1}) \quad m=12, 13, 14.$$
Therefore, we can conclude that $A_2=o_p(1)$.
\end{proof}

\noindent\textbf{Proof of Lemma \ref{clt:hete:4}}
\begin{proof}
By the Theorem 3.2 of \cite{Yin2021spectral}, there exist constants $\mu_{center,4}$, $\mu_{limit,4}$ and $\sigma_{PE}>0$ such that: $$tr(\R^4)-\mu_{center,4}\cvd N(\mu_{limit,4},\sigma_{PE,0}^2).$$ Applying the results in Example 3.2 of \cite{Yin2021spectral} with $g(x)=x^4$, we obtain $\mu_{center,4}=n+\frac{6n^2}{T-1}+\frac{6n^3}{(T-1)^2}+\frac{n^4}{(T-1)^2}$. For the case $g(x)=x^4$ and $\R=\mathbf{I}_n$, the results in Example 3.3 of \cite{Yin2021spectral} shows that $\mu_{limit}=-6c(1+c)^2-2c^2$ and $\sigma_{PE,0}^2=8c^2+96c^3(1+c)^2+16c^2(3c^2+8c+3)^2$. Finally, substituting $c$ with $c_T$ by Slutsky's theorem completes the proof.
\end{proof}

\noindent\textbf{Proof of Lemma \ref{op1:hete:4}}
\begin{proof}
	It is easy to verify that
	$$\begin{aligned}
	tr\left(\hat{\mathbf{R}}^4-\mathbf{R}^4\right)=&2\sum\limits_{i\neq j}\left( \hat{\rho}_{ij}^2-\rho_{ij}^2\right)+2\sum\limits_{i\neq j\neq l}\left( \hat{\rho}_{ij}^2\hat{\rho}_{jl}^2-\rho_{ij}^2\rho_{jl}^2\right)+\sum\limits_{i\neq j}\left( \hat{\rho}_{ij}^4-\rho_{ij}^4\right)\\
	&+\sum\limits_{i\neq j\neq l}\left(\hat{\rho}_{ij}\hat{\rho}_{jl}\hat{\rho}_{il}-\rho_{ij}\rho_{jl}\rho_{il}\right)+\sum\limits_{i\neq j\neq l\neq s} \left(\hat{\rho}_{ij}\hat{\rho}_{jl}\hat{\rho}_{ls}\hat{\rho}_{si}-\rho_{ij}\rho_{jl}\rho_{ls}\rho_{si} \right),\\
  \end{aligned}$$  
where $1\leq i,j,l,s \leq n$. We only aim to show that
\begin{enumerate}

    \item[\textbf{(i)}] $B\stackrel{\Delta}{=}\left|\sum\limits_{i\neq j\neq l}\left( \hat{\rho}_{ij}^2\hat{\rho}_{jl}^2-\rho_{ij}^2\rho_{jl}^2\right)\right|=o_p(1)$,
    \item[\textbf{(ii)}] 
    $C\stackrel{\Delta}{=}\left|\sum\limits_{i\neq j}\left( \hat{\rho}_{ij}^4-\rho_{ij}^4\right)\right|=o_p(1)$,
    \item[\textbf{(iii)}] $D\stackrel{\Delta}{=}\left|\sum\limits_{i\neq j\neq l}\left(\hat{\rho}_{ij}\hat{\rho}_{jl}\hat{\rho}_{il}-\rho_{ij}\rho_{jl}\rho_{il}\right)\right|=o_p(1)$,
    \item[\textbf{(iv)}]  $E\stackrel{\Delta}{=}\left|\sum\limits_{i\neq j\neq l\neq s} \left(\hat{\rho}_{ij}\hat{\rho}_{jl}\hat{\rho}_{ls}\hat{\rho}_{si}-\rho_{ij}\rho_{jl}\rho_{ls}\rho_{si} \right)\right|=o_p(1)$
\end{enumerate}
since  we have $\left|\sum\limits_{i\neq j}\left( \hat{\rho}_{ij}^2-\rho_{ij}^2\right)\right|=o_p(1)$ by Lemma 2.

\noindent\textbf{(i)} By direct calculation we have 
$$\begin{aligned}
  B&=\left|\sum\limits_{i\neq j\neq l}\left( \hat{\rho}_{ij}^2\hat{\rho}_{jl}^2-\rho_{ij}^2\rho_{jl}^2\right)\right| \\&=\left|\sum\limits_{i\neq j\neq l}\rho_{ij}^2\rho_{jl}^2 \left(\frac{\hat{\rho}_{ij}^2\hat{\rho}_{jl}^2}{\rho_{ij}^2\rho_{jl}^2}-1 \right)  \right|
  \\&\le \sqrt{\sum\limits_{i\neq j\neq l}\rho_{ij}^4\rho_{jl}^4} \cdot\sqrt{\sum\limits_{i\neq j\neq l}\left(\frac{\hat{\rho}_{ij}^2\hat{\rho}_{jl}^2}{\rho_{ij}^2\rho_{jl}^2}-1 \right) ^2} \\
 &= \sqrt{\frac1{T^{\alpha_2}}\sum\limits_{i\neq j\neq l}\rho_{ij}^4\rho_{jl}^4}\cdot \sqrt{T^{\alpha_2}\sum\limits_{i\neq j\neq l}\left( \frac{\Stij^2S^2_{T,j,l}\hat{S}^2_{T,j,j}\hat{S}_{T,i,i}\hat{S}_{T,l,l}-\hStij^2\hat{S}^2_{T,j,l}S^2_{T,j,j}S_{T,i,i}S_{T,l,l}}{\Stij^2S^2_{T,j,l}\hat{S}^2_{T,j,j}\hat{S}_{T,i,i}\hat{S}_{T,l,l}} \right)^2},\\
\end{aligned}$$
where constant $0<\alpha_2<1$. We show that:
\begin{enumerate}
    \item[\textbf{(i.1)}] $$B_1\stackrel{\Delta}{=}\left|\frac1{T^{\alpha_2}}\sum\limits_{i\neq j\neq l}\rho_{ij}^4\rho_{jl}^4\right|=o_p(1),$$
    \item[\textbf{(i.2)}] \begin{equation*}B_2\stackrel{\Delta}{=}\left|T^{\alpha_2}\sum\limits_{i\neq j\neq l}\left( \frac{\Stij^2S^2_{T,j,l}\hat{S}^2_{T,j,j}\hat{S}_{T,i,i}\hat{S}_{T,l,l}-\hStij^2\hat{S}^2_{T,j,l}S^2_{T,j,j}S_{T,i,i}S_{T,l,l}}{\Stij^2S^2_{T,j,l}\hat{S}^2_{T,j,j}\hat{S}_{T,i,i}\hat{S}_{T,l,l}} \right)^2\right|=o_p(1).
    \end{equation*} 
\end{enumerate}
\textbf{(i.1)} By lemma \ref{lem:max_rho}, we have 
\begin{align*}
  \left|  \frac1{T^{\alpha_2}}\sum\limits_{i\neq j\neq l}\rho_{ij}^4\rho_{jl}^4 \right| \leq O_p\left(\frac{(\logn)^4}{nT^{\alpha_2}}\right)=o_p(1).
\end{align*}
Therefore, $B_1=o_p(1)$ holds.

\noindent \textbf{(i.2)} By direct calculation, we have 
$$\begin{aligned}
  B_2  \le T^{\alpha_2-16}\sum\limits_{i\neq j\neq l}\frac{|I_{B_2}|^2}{(\Stij^2S^2_{T,j,l}\hat{S}^2_{T,j,j}\hat{S}_{T,i,i}\hat{S}_{T,l,l})^2},
\end{aligned}$$where
$$\begin{aligned}
  I_{B_2}=&(\bv_i'\bv_j)^2(\bv_j'\bv_l)^2\left[(\bv_j'-\bw_j')(\bv_j-\bw_j)\right]^2(\bv_i'-\bw_i')(\bv_i-\bw_i)(\bv_l'-\bw_l')(\bv_l-\bw_l)\\&-\left[(\bv_i'-\bw_i')(\bv_j-\bw_j)\right]^2\left[(\bv_j'-\bw_j')(\bv_l-\bw_l)\right]^2(\bv_j'\bv_j)^2\bv_i'\bv_i\bv_l'\bv_l.\\
\end{aligned}$$
Consequently, 
$$\begin{aligned}
 |B_2|\leq & T^{\alpha_2-16}\sum\limits_{i\neq j\neq l}\frac1{(\Stij^2S^2_{T,j,l}\hat{S}^2_{T,j,j}\hat{S}_{T,i,i}\hat{S}_{T,l,l})^2}\Big|\sum\limits_{m=1}^{153}B_{2,m}\Big|^2\\=& T^{\alpha_2-16}\sum\limits_{i\neq j\neq l}\frac1{(\Stij^2S^2_{T,j,l}\hat{S}^2_{T,j,j}\hat{S}_{T,i,i}\hat{S}_{T,l,l})^2}\Big|\sum\limits_{m_1,m_2}\eta_{m_1,m_2}B_{2,m_1}B_{2,m_2}\Big|,
\end{aligned}$$
 where $\eta_{m_1,m_2}$ is a constant only depending on $m_1$ and $m_2$. By the same arguments in the proof of Lemma 2, we only need to show that $$T^{\alpha_2-16}\sum\limits_{i\neq j\neq l}\frac{|B_{2,m}|^2}{(\Stij^2S^2_{T,j,l}\hat{S}^2_{T,j,j}\hat{S}_{T,i,i}\hat{S}_{T,l,l})^2}=o_p(1),$$ for any $1\leq m \leq 153$. By Lemma \ref{lem:estimates}, for $1\leq m \leq 53$, one can easily show that the stochastic order dominating terms are $$ T^{\alpha_2-16}\sum\limits_{i\neq j\neq l}\frac{\left|-2(\bv_i'\bv_j)^2(\bv_j'\bv_l)^2(\bv_j'\bv_j)^2\bw_i'\bv_i\bv_l'\bv_l\right|^2}{({\Stij^2S^2_{T,j,l}\hat{S}^2_{T,j,j}\hat{S}_{T,i,i}\hat{S}_{T,l,l}})^2},$$$$ T^{\alpha_2-16}\sum\limits_{i\neq j\neq l}\frac{\left|(\bv_i'\bv_j)^2(\bv_j'\bv_l)^2(\bv_j'\bv_j)^2\bw_i'\bw_i\bv_l'\bv_l\right|^2}{({\Stij^2S^2_{T,j,l}\hat{S}^2_{T,j,j}\hat{S}_{T,i,i}\hat{S}_{T,l,l}})^2} ,  $$
$$ T^{\alpha_2-16}\sum\limits_{i\neq j\neq l}\frac{\left|-2(\bv_i'\bv_j)^2(\bv_j'\bv_l)^2(\bv_j'\bv_j)^2\bv_i'\bv_i\bw_l'\bw_l\right|^2}{{({\Stij^2S^2_{T,j,l}\hat{S}^2_{T,j,j}\hat{S}_{T,i,i}\hat{S}_{T,l,l}})^2} } ,$$$$ T^{\alpha_2-16}\sum\limits_{i\neq j\neq l}\frac{\left|(\bv_i'\bv_j)^2(\bv_j'\bv_l)^2(\bv_j'\bv_j)^2\bv_j'\bv_j\bw_l'\bw_l\right|^2}{{(\Stij^2S^2_{T,j,l}\hat{S}^2_{T,j,j}\hat{S}_{T,i,i}\hat{S}_{T,l,l})^2}},$$
which have the same order $O_p\Big(n^{\frac{10}{r_1}+4\epsilon_1+8\epsilon_2+\alpha_2-3}\Big)=o_p(1)$. For $54\leq m \leq 153$, stochastic order dominating terms have the same order of $$T^{\alpha_2-16}\sum\limits_{i\neq j\neq l}\frac{\left|(\bv_i'\bv_j)^2\bv_j'\bv_l\bw_j'\bv_l(\bv_j'\bv_j)^2\bv_i'\bv_i\bv_l'\bv_l\right|^2}{{(\Stij^2S^2_{T,j,l}\hat{S}^2_{T,j,j}\hat{S}_{T,i,i}\hat{S}_{T,l,l})^2}},$$ whose order are $O_p(n^{\frac{8}{r_1}+8\epsilon_1+6\epsilon_2+\alpha_2-1})=o_p(1)$. Therefore, we can conclude that $B_2=o_p(1)$.\\

\noindent\textbf{(ii)} By direct calculation we have 
$$\begin{aligned}
  C&=\left|\sum\limits_{i\neq j}\left( \hat{\rho}_{ij}^4-\rho_{ij}^4\right)\right| \\&=\left|\sum\limits_{i\neq j}\rho_{ij}^4 \left(\frac{\hat{\rho}_{ij}^4}{\rho_{ij}^4}-1 \right)  \right|
  \\& \le \sqrt{\sum\limits_{i\neq j}\rho_{ij}^8} \cdot\sqrt{\sum\limits_{i\neq j}\left(\frac{\hat{\rho}_{ij}^4}{\rho_{ij}^4}-1 \right) ^2} \\
 &= \sqrt{\frac1{T^{\alpha_3}}\sum\limits_{i\neq j}\rho_{ij}^8}\cdot \sqrt{T^{\alpha_3}\sum\limits_{i\neq j}\left( \frac{\hStij^4 S_{T,i,i}^2S_{T,j,j}^2-\Stij^4\hat{S}^2_{T,i,i}\hat{S}^2_{T,j,j}}{\Stij^4\hat{S}^2_{T,i,i}\hat{S}^2_{T,j,j}} \right)^2},\\
\end{aligned}$$
where constant $0<\alpha_3<1$. We show that:
\begin{enumerate}
    \item[\textbf{(ii.1)}] $$C_1\stackrel{\Delta}{=}\left|\frac1{T^{\alpha_3}}\sum\limits_{i\neq j}\rho_{ij}^8\right|=o_p(1),$$
    \item[\textbf{(ii.2)}]  $$C_2\stackrel{\Delta}{=}\left|T^{\alpha_3}\sum\limits_{i\neq j}\left( \frac{\hStij^4 S_{T,i,i}^2S_{T,j,j}^2-\Stij^4\hat{S}^2_{T,i,i}\hat{S}^2_{T,j,j}}{\Stij^4\hat{S}^2_{T,i,i}\hat{S}^2_{T,j,j}} \right)^2\right|=o_p(1).$$
\end{enumerate}
\textbf{(ii.1)} By lemma \ref{lem:max_rho}, we have 
\begin{align*}
 \left|  \frac1{T^{\alpha_3}}\sum\limits_{i\neq j}\rho_{ij}^8 \right| \leq O_p\left(\frac{(\logn)^4}{n^2T^{\alpha_3}}\right)=o_p(1).
\end{align*}
Therefore, $C_1=o_p(1)$ holds.

\noindent \textbf{(ii.2)} By direct calculation, we have 
$$\begin{aligned}
  C_2 \le T^{\alpha_3-16}\sum\limits_{i\neq j}\frac{|I_{C_2}|^2}{\Stij^8\hat{S}^4_{T,i,i}\hat{S}^4_{T,j,j}},
\end{aligned}$$where
$$\begin{aligned}
I_{C_2}=&\left[(\bv_i'-\bw_i')(\bv_j-\bw_j)\right]^4(\bv_i'\bv_i)^2(\bv_j'\bv_j)^2
  &-(\bv_i'\bv_j)^4\left[(\bv_i'-\bw_i')(\bv_i-\bw_i)\right]^2\left[(\bv_j'-\bw_j')(\bv_j-\bw_j)\right]^2.
  \end{aligned}$$
Consequently, 
$$\begin{aligned}
 C_2&\leq T^{\alpha_3-16}\sum\limits_{i\neq j}\frac1{\Stij^8\hat{S}^4_{T,i,i}\hat{S}^4_{T,j,j}}|\sum\limits_{m=1}^{68}C_{2,m}|^2\\&= T^{\alpha_3-16}\sum\limits_{i\neq j}\frac1{\Stij^8\hat{S}^4_{T,i,i}\hat{S}^4_{T,j,j}}|\sum\limits_{m_1,m_2}\delta_{m_1,m_2}C_{2,m_1}C_{2,m_2}|,
\end{aligned}$$
 where $\delta_{m_1,m_2}$ is a constant only depending on $m_1$ and $m_2$. By the same arguments in the proof of Lemma 2, we only need to show that $$ T^{\alpha_3-16}\sum\limits_{i\neq j}\frac{|C_{2,m}|^2}{\Stij^8\hat{S}^4_{T,i,i}\hat{S}^4_{T,j,j}}=o_p(1),$$ for any $1\leq m \leq 68$.
By Lemma \ref{lem:estimates}, for $1\leq m \leq 32$, one can show that the stochastic order dominating terms are $$ T^{\alpha_3-16}\sum\limits_{i\neq j}\frac{\left|-4(\bv_i'\bv_j)^3\bw_i'\bv_j(\bv_i'\bv_i)^2(\bv_j'\bv_j)^2\right|^2}{\Stij^8\hat{S}^4_{T,i,i}\hat{S}^4_{T,j,j}}, $$ $$T^{\alpha_3-16}\sum\limits_{i\neq j}\frac{\left|-4(\bv_i'\bv_j)^3\bv_i'\bw_j(\bv_i'\bv_i)^2(\bv_j'\bv_j)^2\right|^2}{\Stij^8\hat{S}^4_{T,i,i}\hat{S}^4_{T,j,j}} ,  $$ and
$$ T^{\alpha_3-16}\sum\limits_{i\neq j}\frac{\left|-4(\bv_i'\bv_j)^3\bw_i'\bw_j(\bv_i'\bv_i)^2(\bv_j'\bv_j)^2\right|^2}{{\Stij^8\hat{S}^4_{T,i,i}\hat{S}^4_{T,j,j}} } ,$$
which have the same order $O_p(n^{\frac{10}{r_1}+4\epsilon_1+8\epsilon_2+\alpha_3-4})=o_p(1)$. For $33\leq m \leq 68$, stochastic order dominating terms have the same order of $$T^{\alpha_3-16}\sum\limits_{i\neq j}\frac{\left|-4(\bv_i'\bv_j)^4(\bv_i'\bv_i)^2\bv_j'\bv_j\bw_j'\bv_j\right|^2}{\Stij^8\hat{S}^4_{T,i,i}\hat{S}^4_{T,j,j}},$$ whose order is $O_p(n^{\frac{8}{r_1}+8\epsilon_1+6\epsilon_2+\alpha_2-1})=o_p(1)$. Therefore, we can conclude that $C_2=o_p(1)$. \\

\noindent\textbf{(iii)}By direct calculation, we have 
$$\begin{aligned}
  D\stackrel{\Delta}{=}\left|\sum\limits_{i\neq j\neq l}\left(\hat{\rho}_{ij}\hat{\rho}_{jl}\hat{\rho}_{il}-\rho_{ij}\rho_{jl}\rho_{il}\right)\right|& =\left|\sum\limits_{i\neq j\neq l}\rho_{ij}\rho_{jl}\rho_{il} \left(\frac{\hat{\rho}_{ij}\hat{\rho}_{jl}\hat{\rho}_{il}}{\rho_{ij}\rho_{jl}\rho_{il}}-1 \right)  \right|\\
  &\le \sqrt{\sum\limits_{i\neq j\neq l}\rho_{ij}^2\rho_{jl}^2\rho_{il}^2} \cdot\sqrt{\sum\limits_{i\neq j\neq l}\left(\frac{\hat{\rho}_{ij}\hat{\rho}_{jl}\hat{\rho}_{il}}{\rho_{ij}\rho_{jl}\rho_{il}}-1 \right) ^2} .\\
\end{aligned}$$
For constant $0<\alpha_4<1$, we show that:
\begin{enumerate}
    \item[\textbf{(iii.1)}] $$D_1\stackrel{\Delta}{=}\left|\frac1{T^{\alpha_4}}\sum\limits_{i\neq j\neq l}\rho_{ij}^2\rho_{jl}^2\rho_{il}^2\right|=o_p(1).$$
    \item[\textbf{(iii.2)}]  $$D_2\stackrel{\Delta}{=}\left|T^{\alpha_4-12}\sum\limits_{i\neq j\neq l}\left( \frac{\hat{S}_{T,i,j}\hat{S}_{T,j,l}\hat{S}_{T,i,l}S_{T,i,i}S_{T,j,j}S_{T,l,l}-S_{T,i,j}S_{T,j,l}S_{T,i,l}\hat{S}_{T,i,i}\hat{S}_{T,j,j}\hat{S}_{T,l,l}}{S_{T,i,j}S_{T,j,l}S_{T,i,l}\hat{S}_{T,i,i}\hat{S}_{T,j,j}\hat{S}_{T,l,l} }\right)^2\right|=o_p(1).$$
\end{enumerate}
\noindent\textbf{(iii.1)} By lemma \ref{lem:max_rho}, we have
\begin{align*}
    \left|\frac1{T^{\alpha_4}}\sum\limits_{i\neq j\neq l}\rho_{ij}^2\rho_{jl}^2\rho_{il}^2\right|\leq\frac1{T^{\alpha_4}}O_p\left((\frac{\logn}{n})^3\cdot n^3 \right)=O_p\Big(\frac{(\logn)^3}{T^{\alpha_4}}\Big)=o_p(1).
\end{align*}Therefore, $F_1=o_p(1)$ holds.

\noindent\textbf{(iii.2)} By direct calculation, we have 
$$\begin{aligned}
D_2 \le T^{\alpha_4-12}\sum\limits_{i\neq j\neq l}\frac{|I_{D_2}|^2}{(S_{T,i,j}S_{T,j,l}S_{T,i,l}\hat{S}_{T,i,i}\hat{S}_{T,j,j}\hat{S}_{T,l,l})^2},
\end{aligned}$$
where
$$\begin{aligned}
I_{D_2}&=(\bv_i'-\bw_i')(\bv_j-\bw_j)(\bv_j'-\bw_j')(\bv_l-\bw_l)(\bv_i'-\bw_i')(\bv_l-\bw_l)\bv_i'\bv_i\bv_j'\bv_j\bv_l'\bv_l\\
  &-\bv_i'\bv_j\bv_j'\bv_l\bv_i'\bv_l(\bv_i'-\bw_i')(\bv_i-\bw_i)(\bv_j'-\bw_j')(\bv_j-\bw_j)(\bv_l'-\bw_l')(\bv_l-\bw_l).
\end{aligned}$$
Thus
$$\begin{aligned}
 D_2&\leq T^{\alpha_4-12}\sum\limits_{i\neq j\neq l}\frac{1}{(S_{T,i,j}S_{T,j,l}S_{T,i,l}\hat{S}_{T,i,i}\hat{S}_{T,j,j}\hat{S}_{T,l,l})^2}\Big|\sum\limits_{m=1}^{80}D_{2,m}\Big|^2\\
 &=T^{\alpha_4-12}\sum\limits_{i\neq j\neq l}\frac{1}{(S_{T,i,j}S_{T,j,l}S_{T,i,l}\hat{S}_{T,i,i}\hat{S}_{T,j,j}\hat{S}_{T,l,l})^2}\Big|\sum\limits_{m_1,m_2}\xi_{m_1,m_2}D_{2,m_1}D_{2,m_2}\Big|,
\end{aligned}$$
 where $\xi_{m_1,m_2}$ is a constant only depending on $m_1$ and $m_2$. By the same arguments in the proof of Lemma 2, we only need to show that $$T^{\alpha_4-12}\sum\limits_{i\neq j\neq l}\frac{|D_{2,m}|^2}{(S_{T,i,j}S_{T,j,l}S_{T,i,l}\hat{S}_{T,i,i}\hat{S}_{T,j,j}\hat{S}_{T,l,l})^2}=o_p(1),$$ for any $1\leq m \leq 80$.
By Lemma \ref{lem:estimates}, for $1\leq m \leq 63$, one can  show that the stochastic order dominating terms have the same order of  $$T^{\alpha_4-12}\sum\limits_{i\neq j\neq l}\frac{\left|-\bv_i'\bv_j\bv_j'\bv_l\bw_i'\bv_l\bv_i'\bv_i\bv_j'\bv_j\bv_l'\bv_l\right|^2}{{(S_{T,i,j}S_{T,j,l}S_{T,i,l}\hat{S}_{T,i,i}\hat{S}_{T,j,j}\hat{S}_{T,l,l})^2}} ,  $$
whose orders are $O_p(n^{\frac{6}{r_1}+4\epsilon_1+4\epsilon_2+\alpha_4-1})=o_p(1)$. For $64\leq m \leq 80$, the stochastic order dominating terms have the same order of  $$T^{\alpha_4-12}\sum\limits_{i\neq j\neq l}\frac{\left|-\bv_i'\bv_j\bv_j'\bv_l\bv_i'\bv_l\bv_i'\bv_i\bv_j'\bv_j\bv_l'\bw_l\right|^2}{{(S_{T,i,j}S_{T,j,l}S_{T,i,l}\hat{S}_{T,i,i}\hat{S}_{T,j,j}\hat{S}_{T,l,l})^2}} ,  $$
whose orders are $O_p(n^{\frac{8}{r_1}+4\epsilon_1+6\epsilon_2+\alpha_4-2})=o_p(1)$. Therefore, we can conclude that $D_2=o_p(1).$\\

\noindent \textbf{(iv)} By direct calculation, we have 
$$\begin{aligned}
  E\stackrel{\Delta}{=}\left|\sum\limits_{i\neq j\neq l\neq s} \left(\hat{\rho}_{ij}\hat{\rho}_{jl}\hat{\rho}_{ls}\hat{\rho}_{si}-\rho_{ij}\rho_{jl}\rho_{ls}\rho_{si} \right)\right|& =\left|\sum\limits_{i\neq j\neq l\neq s}\rho_{ij}\rho_{jl}\rho_{ls}\rho_{si} \left(\frac{\hat{\rho}_{ij}\hat{\rho}_{jl}\hat{\rho}_{ls}\hat{\rho}_{si}}{\rho_{ij}\rho_{jl}\rho_{ls}\rho_{si}}-1 \right)  \right|\\
  &\le \sqrt{\sum\limits_{i\neq j\neq l\neq s}\rho_{ij}^2\rho_{jl}^2\rho_{ls}^2\rho_{si}^2} \cdot\sqrt{\sum\limits_{i\neq j\neq l\neq s}\left(\frac{\hat{\rho}_{ij}\hat{\rho}_{jl}\hat{\rho}_{ls}\hat{\rho}_{si}}{\rho_{ij}\rho_{jl}\rho_{ls}\rho_{si}}-1 \right) ^2}. \\
\end{aligned}$$
For constant $0<\alpha_5<1$, we show that:
\begin{enumerate}
    \item[\textbf{(iv.1)}] $$E_1\stackrel{\Delta}{=}\left|\frac1{T^{\alpha_5}}\sum\limits_{i\neq j\neq l\neq s}\rho_{ij}^2\rho_{jl}^2\rho_{ls}^2\rho_{si}^2\right|=o_p(1),$$
    \item[\textbf{(iv.2)}] 
 $$E_2\stackrel{\Delta}{=}\left|T^{\alpha_5}\sum\limits_{i\neq j\neq l\neq s}\left(\frac{\hat{\rho}_{ij}\hat{\rho}_{jl}\hat{\rho}_{ls}\hat{\rho}_{si}}{\rho_{ij}\rho_{jl}\rho_{ls}\rho_{si}}-1 \right)^2 \right|=o_p(1).$$

\end{enumerate}
\textbf{(iv.1)} By lemma \ref{lem:max_rho}, we have
\begin{align*}
   \left|\frac1{T^{\alpha_5}}\sum\limits_{i\neq j\neq l\neq s}\rho_{ij}^2\rho_{jl}^2\rho_{ls}^2\rho_{si}^2\right|\leq \frac1{T^{\alpha_5}}O_p\left((\frac{\logn}{n})^4\cdot n^4\right)=O_p\left(\frac{(\logn)^4}{T^{\alpha_5}}\right)=o_p(1).
\end{align*}
Therefore, $E_1=o_p(1)$ holds.

\textbf{(ii.2)} By direct calculation, we have 
$$\begin{aligned}
  \left|E_2 \right| \\=&\left|T^{\alpha_5-16}\sum\limits_{i\neq j\neq l\neq s}\left( \frac{\hat{S}_{T,i,j}\hat{S}_{T,j,l}\hat{S}_{T,l,s}\hat{S}_{T,s,i}S_{T,i,i}S_{T,j,j}S_{T,l,l}S_{T,s,s}-S_{T,i,j}S_{T,j,l}S_{T,l,s}S_{T,s,i}\hat{S}_{T,i,i}\hat{S}_{T,j,j}\hat{S}_{T,l,l}\hat{S}_{T,s,s}}{S_{T,i,j}S_{T,j,l}S_{T,l,s}S_{T,s,i}\hat{S}_{T,i,i}\hat{S}_{T,j,j}\hat{S}_{T,l,l}\hat{S}_{T,s,s}} \right)^2\right|\\\le& T^{\alpha_5-16}\sum\limits_{i\neq j\neq l\neq s}\frac{|I_{E_2}|^2}{(S_{T,i,j}S_{T,j,l}S_{T,l,s}S_{T,s,i}\hat{S}_{T,i,i}\hat{S}_{T,j,j}\hat{S}_{T,l,l}\hat{S}_{T,s,s})^2},
\end{aligned}$$
where
$$\begin{aligned}
I_{E_2}&=(\bv_i'-\bw_i')(\bv_j-\bw_j)(\bv_j'-\bw_j')(\bv_l-\bw_l)(\bv_l'-\bw_l')(\bv_s-\bw_s)(\bv_s'-\bw_s')(\bv_i-\bw_i)\bv_i'\bv_i\bv_j'\bv_j\bv_l'\bv_l\bv_s'\bv_s\\
  &-\bv_i'\bv_j\bv_j'\bv_l\bv_l'\bv_s\bv_s'\bv_i(\bv_i'-\bw_i')(\bv_i-\bw_i)(\bv_j'-\bw_j')(\bv_j-\bw_j)(\bv_l'-\bw_l')(\bv_l-\bw_l)(\bv_s'-\bw_s')(\bv_s-\bw_s).
\end{aligned}$$
Thus
$$\begin{aligned}
 E_2&\leq  T^{\alpha_5-16}\sum\limits_{i\neq j\neq l\neq s}\frac1{(S_{T,i,j}S_{T,j,l}S_{T,l,s}S_{T,s,i}\hat{S}_{T,i,i}\hat{S}_{T,j,j}\hat{S}_{T,l,l}\hat{S}_{T,s,s})^2}\Big|\sum\limits_{m=1}^{336}E_{2,m}\Big|^2\\
 &=  T^{\alpha_5-16}\sum\limits_{i\neq j\neq l\neq s}\frac1{(S_{T,i,j}S_{T,j,l}S_{T,l,s}S_{T,s,i}\hat{S}_{T,i,i}\hat{S}_{T,j,j}\hat{S}_{T,l,l}\hat{S}_{T,s,s})^2}\Big|\sum\limits_{m_1,m_2}\lambda_{m_1,m_2}E_{2,m_1}E_{2,m_1}\Big|,
\end{aligned}$$
 where $\lambda_{m_1,m_2}$ is a constant only depending on $m_1$ and $m_2$. By the same arguments in the proof of Lemma 2, we only need to show that $$ T^{\alpha_5-16}\sum\limits_{i\neq j\neq l\neq s}\frac{|E_{2,m}|^2}{(S_{T,i,j}S_{T,j,l}S_{T,l,s}S_{T,s,i}\hat{S}_{T,i,i}\hat{S}_{T,j,j}\hat{S}_{T,l,l}\hat{S}_{T,s,s})^2}=o_p(1),$$ for any $1\leq m \leq 336$. By Lemma \ref{lem:estimates}, for $1\leq m \leq 255$, stochastic order dominating terms have the same order of $$ T^{\alpha_5-16}\sum\limits_{i\neq j\neq l\neq s}\frac{\left|\bv_i'\bv_j\bv_j'\bv_l\bv_l'\bv_s\bv_s'\bw_i\bv_i'\bv_i\bv_j'\bv_j\bv_l'\bv_l\bv_s'\bv_s\right|^2}{{(S_{T,i,j}S_{T,j,l}S_{T,l,s}S_{T,s,i}\hat{S}_{T,i,i}\hat{S}_{T,j,j}\hat{S}_{T,l,l}\hat{S}_{T,s,s})^2}},$$
whose orders are $O_p(n^{\frac{8}{r_1}+4\epsilon_1+6\epsilon_2+\alpha_5-1})=o_p(1)$.
For $256\leq m \leq 336$, one can  show that the stochastic order dominating terms have the same order of $$ T^{\alpha_5-16}\sum\limits_{i\neq j\neq l\neq s}\frac{\left|-2\bv_i'\bv_j\bv_j'\bv_l\bv_l'\bv_s\bv_s'\bv_i\bv_i'\bv_i\bv_j'\bv_j\bv_l'\bv_l\bv_s'\bw_s\right|^2}{{(S_{T,i,j}S_{T,j,l}S_{T,l,s}S_{T,s,i}\hat{S}_{T,i,i}\hat{S}_{T,j,j}\hat{S}_{T,l,l}\hat{S}_{T,s,s})^2}} ,  $$
whose orders are $O_p(n^{\frac{10}{r_1}+4\epsilon_1+8\epsilon_2+\alpha_5-2})=o_p(1)$. Therefore, we can conclude that $E_2=o_p(1)$. Finally, proof of proposition 2 is completed.
\end{proof}

\noindent\textbf{Proof of Theorem 3 and 4}

For the dynamic panel data model, let $\Z_i=(\z_{i1},\dots,\z_{iT})$, then  $\hat{\pphi}_i=\left(\Z_i\Z_i^{\prime}\right)^{-1}\Z_i\Y_i$ is the OLS estimator and the residuals are given by $\hat{\vv}_{it}=v_{it}-\z_{it}^{\prime}\left(\hat{\pphi}_i-\pphi_i\right)$. In vector form, $\hat{\bv}_i=\bv_i-\Z_i^{\prime}\left(\hat{\pphi}_i-\pphi_i\right)$. Define $\hat{\hat{\mathbf{V}}}=\begin{pmatrix} \hat{\bv}_{1} , \cdots,  \hat{\bv}_{n} \end{pmatrix}$, $\hat{\hat{\bw}}_i=\Z_i^{\prime}\left(\hat{\pphi}_i-\pphi_i\right)$ and $\hat{\hat{\mathbf{W}}}=\begin{pmatrix} \hat{\hat{\bw}}_{1} , \cdots,  \hat{\hat{\bw}}_{n} \end{pmatrix}$. Using this notation,  $\hat{\bv}_i=\bv_i-\hat{\hat{\bw}}_i$. Replacing $\hat{\bbeta}_i$ and $\X_i$ with $\hat{\pphi}_i$ and $\Z_i$, respectively, the proofs of Theorem  3 and 4 follow along the same arguments above, that is, we only need to verify that (a) and (b) in Lemma \ref{lem:estimates} still hold for the dynamic panel data model.
\begin{lemma}\label{lem:estimates:dynamic}
  Under  
  Assumptions 1, 2, 3, 4 and 5, for any $\epsilon>0$ and some integer $r_1\geq 3$, i.e. $E|v_{it}|^{r_1}<\infty$
  \begin{enumerate}
  \item[(a)\ ]  $\displaystyle \max_{1\le i ,  j \le n}    |\inn{\bv_i}{\hat{\hat{\bw}}_j}|=O_p (n^{\frac{1}{r_1}+2\epsilon_1})$.
  \item[(b)\ ]  $\displaystyle \max_{1\le i, j \le n}    |\inn{\hat{\hat{\bw}}_i}{\hat{\hat{\bw}}_j}|=O_p (n^{\frac{1}{r_1}+2\epsilon_1})$.
 
  \end{enumerate}
  \end{lemma}
  \begin{proof}
(a). Firstly, for the case $i=j$
$$\begin{aligned}
    \max_{1\leq i \leq n}|\inn{\bv_i}{\hat{\hat{\bw}}_i}|= \max_{1\leq i \leq n}|\bv_i'\Z_i'(\hat{\pphi}_i-\pphi_i)|\leq \max_{1\leq i \leq n}\|\bv_i'\Z_i'\|\max_{1\leq i \leq n}\|\hat{\pphi}_i-\pphi_i\|.
\end{aligned}$$
We have $\max_{1\leq i \leq n}\|\hat{\pphi}_i-\pphi_i\|=O_p(n^{\frac{1}{2r_1}+\epsilon_1-\frac{1}{2}})$ for some integer $r_3>3$ and $\epsilon_3>0$ by Lemma \ref{lem:estimates}, then $$\max_{1\leq i \leq n}\|\bv_i'\Z_i'\|=\max_{1\leq i \leq n}\|\bv_i'(\Y_i,\X_i')\|=\max_{1\leq i \leq n}\|(\bv_i'\Y_i,\bv_i'\X_i')\|\leq \max_{1\leq i \leq n}\|\bv_i'\Y_i\|+\max_{1\leq i \leq n}\|\bv_i'\X_i'\|.$$
For $\max_{1\leq i \leq n}\|\bv_i'\X_i'\|$: $$\max_{1\leq i \leq n}\|\bv_i'\X_i'\|=\max_{1\leq i \leq n}\|\sqrt{T}\boldsymbol{\xi}_i'(\frac{\X_i\X_i'}{T})^{\frac{1}{2}}\|\leq\max_{1\leq i \leq n}\|\sqrt{T}\boldsymbol{\xi}_i'\|\cdot O_p(1)=O_p(n^{\frac{1}{2r_1}+\epsilon_1+\frac{1}{2}}).$$
For $\max_{1\leq i \leq n}\|\bv_i'\Y_i'\|$: Applying martingale theory, we show that $\bv_i'\Y_i'=\sum_{t=1}^Ty_{it-1}v_{it}$ converged to a centered normal distribution. Let $\mathcal{S}_T=\sum_{t=1}^Ty_{it-1}v_{it}$, we aim to verify conditions A1 and A2 imposed in Corollary 2.1.10 of \cite{duflo2013random}. Firstly,  $\langle M \rangle_T\stackrel{\Delta}{=}\sum_{t=1}^TE\left((\mathcal{S}_t-\mathcal{S}_{t-1})^2|\mathcal{F}_{t-1}\right)=\sum_{t=1}^TE\left(y_{it-1}^2v_{it}^2|\mathcal{F}_{t-1}\right)=\sigma^2\sum_{t=1}^Ty_{it-1}^2,$ where $\mathcal{F}_{t-1}$ is the corresponding filtration. Therefore,  $\langle M \rangle_T /T \cvp \sigma^2E(y_{i,0}^2)$  under Assumption 5(i), so that A1 holds. Note that the Lyapunov condition $$\frac{1}{(\sqrt{T})^4}\sum_{t=1}^TE\left(|y_{i,t-1}v_{it}|^4|\mathcal{F}_{t-1}\right)=\frac{E(v_{it})^4}{T^2}\sum_{t=1}^Ty_{it-1}^4\cvp 0$$ holds under Assumption 5(i), which indicates that A2 holds as well. The assertion follows from Corollary 2.1.10 of \cite{duflo2013random}, so that $\max_{1\leq i \leq n}\|\bv_i'\Y_i'\|=O_p(n^{\frac{1}{2r_1}+\epsilon_1})$ by Lemma \ref{lem:estimates}. Consequently, $$ \max_{1\leq i \leq n}|\inn{\bv_i}{\hat{\hat{\bw}}_i}|\leq O_p(n^{\frac{1}{2r_1}+\epsilon_1+\frac{1}{2}})\cdot O_p(n^{\frac{1}{2r_1}+\epsilon_1-\frac{1}{2}})=O_p(n^{\frac{1}{r_1}+2\epsilon_1}).$$ The case $i\ne j$ is similar.

(b)
$$\begin{aligned}
    \max_{1\le i, j \le n}    |\inn{\hat{\hat{\bw}}_i}{\hat{\hat{\bw}}_j}| = \max_{1\le i, j \le n} |(\hat{\pphi}_i-\pphi_i)'\Z_i\Z_j'(\hat{\pphi}_j-\pphi_j)|&\leq \left( \max_i\|\hat{\pphi}_i-\pphi_i\| \max_i\|\Z_i\|\right)^2\\&\leq\left( \max_i\|\hat{\pphi}_i-\pphi_i\| (\max_i\|\Y_i\|+\max_i\|\X_i\|)\right)^2.
\end{aligned}$$
Under Assumption 4(ii) and 5(ii), $\max_i\|\X_i\|\leq \max_i\|\X_i\X_i'\|^{\frac{1}{2}}=O_p(\sqrt{T})$ and $\max_i\|\Y_i\|\leq\max_i\|\Y_i\Y_i'\|^{\frac{1}{2}}=O_p(\sqrt{T})$, so that $ \max_{1\le i, j \le n}    |\inn{\hat{\hat{\bw}}_i}{\hat{\hat{\bw}}_j}|\leq \left( O_p(n^{\frac{1}{2r_1}+\epsilon_1-\frac{1}{2}})\cdot O_p(\sqrt{T})\right)^2= O_p(n^{\frac{1}{r_1}+2\epsilon_1}).$
  \end{proof}

\noindent\textbf{Proofs of Theorem 5 and 6} 

For the fixed effects panel data model, $\hat{\bbeta}$ is the within estimator and the within residuals are given by $\hat{v}_{it}^{fixed}=\tilde{y}_{it}-\tilde{\x}_{it}'(\hat{\bbeta}-\bbeta)$. Let $\Bar{\x}_{i\cdot}=\frac{1}{T}\sum_{t=1}^T\x_{it}$, $\Bar{y}_{i\cdot}=\frac{1}{T}\sum_{t=1}^Ty_{it}$, $\Bar{v}_{i\cdot}=\frac{1}{T}\sum_{t=1}^Tv_{it}$ and $\tilde{v}_{it}=v_{it}-\Bar{v}_{i\cdot}$. Define $\tilde{\bv}_i=(\tilde{v}_{i1},\dots,\tilde{v}_{iT})'$, $\Bar{\bv}_i=(\Bar{v}_{i\cdot}, \dots, \Bar{v}_{i\cdot})'$,  $\tilde{\X}_i=(\tilde{\x}_{i1},\dots,\tilde{\x}_{iT})'$ and $\Bar{\X}_i=(\Bar{\x}_{i\cdot}, \dots, \Bar{\x}_{i\cdot})'$. Let  $\tilde{\bw}_i=\tilde{\X}_i^{\prime}\left(\hat{\bbeta}-\bbeta\right)$ and $\tilde{\mathbf{W}}=\begin{pmatrix} \tilde{\bw}_{1} , \cdots,  \tilde{\bw}_{n} \end{pmatrix}$.
Again, it suffices to verify that Lemma \ref{lem:estimates} (a) and (b) still hold for the fixed effect panel data model.
\begin{lemma}
  Under  
  Assumptions 1, 2,3 and 4, for any $\epsilon>0$ and some integer $r_1\geq 3$, i.e. $E|v_{it}|^{r_1}<\infty$
  \begin{enumerate}
  \item[(a)\ ]  $\displaystyle \max_{1\le i ,  j \le n}    |\inn{\tilde{ \bv}_i}{\tilde{\bw
  }_j}|=O_p (n^{\frac{1}{r_1}+2\epsilon_1})$.
  \item[(b)\ ]  $\displaystyle \max_{1\le i, j \le n}    |\inn{\tilde{\bw
  }_i}{\tilde{\bw
  }_j}|=O_p (n^{\frac{1}{r_1}+2\epsilon_1})$.
 
  \end{enumerate}
  \end{lemma}
\begin{proof} 
(a) When $i=j$, by $\tilde{\bv}_i=\bv_i-\Bar{\bv}_i$ and $\tilde{\X}_i=\X_i-\Bar{\X}_i$
$$\begin{aligned}
    \max_{1\le i \le n}|\inn{\tilde{\bv}_i}{\tilde{\bw
  }_i}|=\max_{1\le i \le n}|\tilde{\bv}_i'\tilde{\X}_i(\hat{\bbeta}-\bbeta)|\le & \max_{1\le i \le n}|\bv_i'\X_i(\hat{\bbeta}-\bbeta)|+ \max_{1\le i \le n}|\Bar{\bv}_i'\X_i(\hat{\bbeta}-\bbeta)|\\&+ \max_{1\le i \le n}|\bv_i'\Bar{\X}_i(\hat{\bbeta}-\bbeta)|+ \max_{1\le i \le n}|\Bar{\bv}_i'\Bar{\X}_i(\hat{\bbeta}-\bbeta)|.
\end{aligned}$$
For $\max_{1\le i \le n}|\bv_i'\X_i(\hat{\bbeta}-\bbeta)|:$
 \begin{align*}
   \max_{1\le i \le n}|\bv_i'\X_i(\hat{\bbeta}-\bbeta)|=\max_{1\le i  \le n}\left|\sum\limits_t v_{it}\x_{it}^\prime\left(\hab-\bbeta\right)\right|
    &= \max_{1\le i  \le n}\left|\sqrt T\boldsymbol{\xi}_i^\prime\left(\frac1T\X_i\X_i^\prime\right)^{\frac12}\left(\hab-\bbeta\right)\right|\\
    &\leq \sqrt T\left|\left|\left(\frac1T\X_i\X_i^\prime\right)^{\frac12}\right|\right|\cdot\max_{1\le i  \le n}\left\|\boldsymbol{\xi}_i\right\|\cdot\left\|\hab-\bbeta\right\|\\
    &=O_p\left(\sqrt T\cdot\frac{n^{\frac{1}{2r_1}+\epsilon_1}}{\sqrt {nT}}\right)=O_p(n^{\frac{1}{2r_1}+\epsilon_1-1/2}).
  \end{align*}
For $\max_{1\le i \le n}|\Bar{\bv}_i'\X_i(\hat{\bbeta}-\bbeta)|:$
\begin{align*}
\frac{1}{T}\Bar{\bv}_i'\X_i=\frac{1}{T}\sum_{t=1}^T\Bar{v}_{i\cdot}\x_{it}=\Big(\frac{1}{T}\sum_{s=1}^Tv_{is}\Big)\Big(\frac{1}{T}\sum_{t=1}^T\x_{it}\Big)= O_p(n^{{2r_1+\epsilon_1-\frac{1}{2}}})\cdot O_p(1)=O_p(n^{2r_1+\epsilon_1-\frac{1}{2}}),
  \end{align*}
uniformly in $i$ since $\frac{1}{T}\sum_{t=1}^Tx_{it}=O_p(1)$ holds uniformly by Assumption 3 and  by  Assumption 4 and Lemma \ref{lem:estimates}, we have $\max_{1\le i \le n}|\frac{1}{T}\sum_{t=1}^Tv_{it}|=O_p(n^{{2r_1+\epsilon_1-\frac{1}{2}}})$. Therefore $\max_{1\le i \le n}|\Bar{\bv}_i'\X_i(\hat{\bbeta}-\bbeta)|=O_p(n^{\frac{1}{2r_1}+\epsilon_1-1/2})$. Using same techniques,  $\max_{1\le i \le n}|\bv_i'\Bar{\X}_i(\hat{\bbeta}-\bbeta)|=O_p(n^{\frac{1}{2r_1}+\epsilon_1-1/2})$ and $\max_{1\le i \le n}|\Bar{\bv}_i'\Bar{\X}_i(\hat{\bbeta}-\bbeta)|=O_p(n^{\frac{1}{2r_1}+\epsilon_1-1/2})$, so that $\max_{1\le i \le n}|\bv_i'\X_i(\hat{\bbeta}-\bbeta)|=O_p(n^{\frac{1}{2r_1}+\epsilon_1-1/2})\le O_p (n^{\frac{1}{r_1}+2\epsilon_1}) $.

The calculations for $i\ne j$ case is similar.

(b)
$$\begin{aligned}
    |\inn{\tilde{\bw}_i}{\tilde{\bw
  }_j}|=|(\hat{\bbeta}-\bbeta)'\tilde{\X}_i'\tilde{\X}_j(\hat{\bbeta}-\bbeta)|\le & |(\hat{\bbeta}-\bbeta)'\X_i'\X_j(\hat{\bbeta}-\bbeta)|+ |(\hat{\bbeta}-\bbeta)'\Bar{\X}_i'\X_j(\hat{\bbeta}-\bbeta)|+\\& |(\hat{\bbeta}-\bbeta)'\X_i'\Bar{\X}_j(\hat{\bbeta}-\bbeta)|+ |(\hat{\bbeta}-\bbeta)'\Bar{\X}_i'\Bar{\X}_j(\hat{\bbeta}-\bbeta)|.
\end{aligned}$$
For $\max_{1\le i \le n}|(\hat{\bbeta}-\bbeta)'\X_i'\X_j(\hat{\bbeta}-\bbeta)|:$
  \begin{align*}
   \| \X_i'    (\hab-\bbeta)\| \| \X_j'    (\hab-\bbeta)\|
     \le  \| \X_i\| \|\X_j\| \|\hab-\bbeta\|^2 
    &\le  \| \X_i\X_i'\|^{1/2} \|\X_j\X_j'\|^{1/2} \|\hab-\bbeta\|^2 \\
    & =  O_p(T^{1/2})  O_p(T^{1/2}) O_p(1/(nT))  = O_p(n^{-1}).
  \end{align*}\
Lastly,
$|(\hat{\bbeta}-\bbeta)'\Bar{\X}_i'\X_j(\hat{\bbeta}-\bbeta)|$, $|(\hat{\bbeta}-\bbeta)'\X_i'\Bar{\X}_j(\hat{\bbeta}-\bbeta)|$ and $|(\hat{\bbeta}-\bbeta)'\Bar{\X}_i'\Bar{\X}_j(\hat{\bbeta}-\bbeta)|$ are all $O_p(n^{-1})$ by Assumption 4. Therefore, $ \max_{1\le i, j \le n}    |\inn{\tilde{W
  }_i}{\tilde{W
  }_j}|=O_p(n^{-1})\le O_p (n^{\frac{1}{r_1}+2\epsilon_1}).$

\end{proof}

\bibliographystyle{chicago}

\bibliography{main}

\end{document}